\documentclass[11pt]{article}
\usepackage{amssymb}

\evensidemargin\oddsidemargin \topmargin2cm
\newtheorem{theorem}{Theorem}
\newtheorem{lemma}{Lemma}
\newtheorem{proposition}{Proposition}

\newtheorem{corollary}[theorem]{Corollary}
\newtheorem{remark}{Remark}
\newenvironment{proof}[1][Proof]{\textbf{#1.} }{\ \rule{0.5em}{0.5em} \smallskip}

\begin{document}

\vspace{2cm}

\begin{center}
{\LARGE Towards the saturation of the Froissart bound}

\vspace{1cm}

{\large Joachim Kupsch }\footnote{%
e-mail: kupsch@physik.uni-kl.de}

{\large Fachbereich Physik, TU Kaiserslautern\\[0pt]
D-67653 Kaiserslautern, Germany}
\end{center}

\vspace{0.5cm}

{\small It is the aim of this paper to review the constructions of pion-pion
scattering amplitudes that rigorously satisfy Mandelstam
analyticity, crossing symmetry, and (at least partly) the
constraints imposed by elastic and inelastic unitarity. Three types
of amplitudes are considered in detail: amplitudes that are given by
a Mandelstam representation, analytic function defined by an
explicit Regge type ansatz, and amplitudes with Regge poles in the
Khuri or the Watson-Sommerfeld representation. The results are
discussed under particular emphasis of a strong increase of the
absorptive part of the forward amplitude and the saturation of the
Froissart bound. Demanding all constraints the optimal construction
obtained so far yields (via the optical theorem) a total cross
section, which decreases like $(\log E)^{-3}$, where $E$ is the
energy of the scattering process. The increasing cross section of
the Froissart bound has been saturated by amplitudes, which satisfy
analyticity, crossing symmetry and the constraints imposed by
inelastic unitarity; but elastic unitarity is missing. The problems
caused by elastic unitarity are discussed in detail.}


\vspace{0.2cm}

\section{Introduction}

One of the outstanding results of the analytic S-matrix theory is the
Froissart bound%
\begin{equation}
\sigma _{tot}(s)\leq const\,\left( \log s\right) ^{2}  \label{i1}
\end{equation}%
for the total cross section of a two particle scattering process, where $s$
is the square of the centre of mass energy. This bound has been derived 1961
by Froissart \cite{Froissart:1961} assuming that the two particle scattering
amplitude has uniformly bounded partial wave amplitudes and satisfies a
Mandelstam representation with a finite number of subtractions. Then Martin
\cite{Martin:1963,Martin:1966} has established this bound using only the
analyticity domain of axiomatic quantum field theory and positivity
properties of the absorptive part. In the meantime experimental results \cite%
{UA4:1993} indicate an increase of the total $pp$ cross section, which is
compatible with a $\left( \log s\right) ^{2}$ behaviour, and future
experiments at BNL-RHIC and CERN-LHC may confirm this increase \cite%
{Bourrely:2003}. The derivation of the Froissart bound (\ref{i1}) follows
from only a part of the analyticity and unitarity properties, which can be
formulated with the elastic two particle scattering amplitude; and it is
still an open problem, whether the Froissart bound can be improved, if all
these constraints are taken into account. For the scattering process with
the strongest crossing restrictions -- the scattering of neutral (or
isospin-1) pions -- the existence of amplitudes, which satisfy elastic
unitarity and the unitarity inequalities in the inelastic regime, has been
derived in 1968 by Atkinson \cite{Atkinson:1968a,Atkinson:1968b}. Thereby
elastic unitarity is incorporated by a non-linear fixed point mapping. The
final construction using a Mandelstam representation with one subtraction
has lead to an amplitude which allows an asymptotic behaviour
\begin{equation}
\sigma _{tot}(s)\sim \left( \log s\right) ^{-3}\;\mathrm{for}\;s\rightarrow
\infty  \label{i2}
\end{equation}%
of the total cross section \cite{Atkinson:1970}. So far there is no proof of
the existence of an amplitude, which comes closer to the Froissart bound, if
all constraints are rigorously fulfilled. But if one does not demand elastic
unitarity, amplitudes have been constructed, which saturate the Froissart
bound \cite{Kupsch:1982}.

In the main part of this paper we give a review about the construction of
amplitudes that rigorously satisfy Mandelstam analyticity, the crossing
symmetry of neutral pions and -- at least to some extend -- the unitarity
constraints, which can be formulated with the elastic scattering amplitude
alone. The discussion concentrates on scattering amplitudes with an optimal
increase of the absorptive part of the forward amplitude at high energies.
In addition to that we investigate the problems, which originate from
elastic unitarity, if one wants to obtain a better result than (\ref{i2}).
All results can be easily extended to the scattering of isospin-1 pions.

The paper is organized as follows. The notations and the assumptions about
analyticity, crossing and unitarity are given in section \ref{not}. In
section \ref{bounds} we recapitulate a few results on upper bounds, which
are important for the subsequent sections. Then we discuss in some detail
three types of amplitudes, which have been investigated in the literature:

\begin{enumerate}
\item[(i)] Amplitudes represented by the Mandelstam spectral integrals%
\newline
A non-linear fixed-point mapping has been developed to incorporate elastic
unitarity for amplitudes of this type \cite%
{Atkinson:1968a,Atkinson:1968b,Atkinson:1969,Atkinson:1970,Kupsch:1969b}. If
the Mandelstam representation has at most one subtraction and the spectral
functions are positive, it is possible to incorporate also the inelastic
unitarity constraints. The result (\ref{i2}) has been obtained for an
amplitude with one subtraction. The construction of these amplitudes and the
limitation to (\ref{i2}) are presented in section \ref{Mandelstam}.

\item[(ii)] Analytic functions with a Regge type asymptotics\newline
Starting with an explicit ansatz one can obtain analytic functions with a
Regge type asymptotics that satisfy Mandelstam analyticity, crossing
symmetry and the constraints of inelastic unitarity. The first solutions
have been obtained for amplitudes with simple Regge poles or double poles
\cite{Kupsch:1971}. Then the total cross section behaves like%
\begin{equation}
\sigma _{tot}(s)\sim s^{\alpha (0)-1}\left( \log s\right) ^{n}\;\mathrm{for}%
\;s\rightarrow \infty  \label{i3}
\end{equation}%
with an intercept $\alpha (0)\leq 1$ of the leading Regge trajectory and an
exponent $n=0$ or $1$ of the logarithm. The validity of the constraints of
inelastic unitarity is derived using a linearization of the quadratic
unitarity inequalities. This technique of linearization is recapitulated in
section \ref{lin} and the construction of the amplitudes with Regge poles is
presented in section \ref{R-inel}. \newline
These methods have been extended in \cite{Kupsch:1982} to amplitudes, which
have crossing Regge cuts and which saturate the Froissart bound. The main
step for this construction is indicated in section \ref{Froissart}.
\end{enumerate}

The amplitudes of section \ref{Regge} do not satisfy elastic unitarity.
Already in 1960 Gribov \cite{Gribov:1960,Gribov:1961} realized that
Mandelstam analyticity, crossing symmetry and elastic unitarity impose
irritating constraints on the high energy behaviour of scattering
amplitudes. The construction -- or the mere proof of existence -- of Regge
amplitudes, which satisfy crossing symmetry and elastic unitarity is
therefore a non-trivial task. For that purpose we consider in section \ref%
{R-el}

\begin{enumerate}
\item[(iii)] Amplitudes in the Khuri representation\newline
To impose elastic unitarity one can generalize the non-linear fixed point
mapping to Regge amplitudes using the Watson-Sommerfeld transform \cite%
{AFJK:1976} or the Khuri representation \cite{Kupsch:1977}. It is possible
to obtain solutions with one rising Regge trajectory that satisfy crossing
symmetry and elastic unitarity and that have an asymptotically constant
total cross-section, see section \ref{R-fixedpoint}. The partial wave
amplitudes of these solutions are uniformly bounded for all energies, but
the validity of the additional constraints of inelastic unitarity remains
questionable. In section \ref{R-unitary} we discuss the present status of
this problem.
\end{enumerate}

As already mentioned the optimal result, which has been obtained assuming
all crossing and unitarity constraints, is (\ref{i2}). Extending the methods
presented in sections \ref{R-inel} and \ref{R-fixedpoint} one might succeed
satisfying all constraints to construct an amplitude with an asymptotically
constant total cross section. But to obtain solutions with increasing total
cross sections one has to develop new techniques. The puzzle of the
Froissart bound is not yet settled.

The problem of the saturation of the Froissart bound has also been
investigated assuming crossing symmetry in the smaller analyticity domain of
the axiomatic quantum field theory \cite{Khuri:1976}. Unfortunately these
methods have not lead to a conclusive answer and will not be discussed in
this paper.

The details of the constructions of the amplitudes can be found in the cited
literature. But for completeness some calculations are given in the
Appendices \ref{PW} -- \ref{Khuri}.

\section{Notations and basic assumptions\label{not}}

The kinematics is always the kinematics of pion scattering with the
Mandelstam variables $s,\,t$ and $u$. The mass of the particles is
normalized to unity such that $s+t+u=4$. We consider the construction of
amplitudes with the following properties:

\begin{enumerate}
\item[1.] The amplitude $A(s,t)$ is holomorphic in the Mandelstam domain%
\begin{equation}
\mathbb{C}_{cut}^{2}=\left\{ (s,t)\in \mathbb{C}^{2}\mid s\notin \left[
4,\infty \right) ,\,t\notin \left[ 4,\infty \right) ,\,s+t\notin \left(
-\infty ,0\right] \right\} .  \label{int1}
\end{equation}

\item[2.] The amplitude $A(s,t)$ is polynomially bounded
\begin{equation}
\left\vert A(s,t)\right\vert \leq const\,\left( 1+\left\vert s\right\vert
+\left\vert t\right\vert \right) ^{n}  \label{int2}
\end{equation}%
for $(s,t)\in \mathbb{C}_{cut}^{2}$ with some positive number $n$, and the
boundary values of $A(s,t)$ are H\"{o}lder continuous with an index $\mu \in
\left( 0,\frac{1}{2}\right] $.

\item[3.] If not stated otherwise, we consider neutral pions, and crossing
symmetry means: \newline
The amplitude is symmetric in the variables $s,\,t$ and $u$%
\begin{equation}
A(s,t)=A(t,s)=A(s,u)=A(u,s)=A(u,t)=A(t,u).  \label{int3}
\end{equation}
\end{enumerate}

A more precise form of Mandelstam analyticity and crossing symmetry is the
following. Let $F(s,t)$ be a polynomially bounded function which is analytic
in the domain%
\begin{equation}
\mathbb{D}=\left\{ (s,t)\in \mathbb{C}^{2}\mid s\notin \left[ 4,\infty
\right) ,\,t\notin \left[ 16,\infty \right) \right\} ,  \label{int11}
\end{equation}%
then $A(s,t)=Sym\,F(s,t)$ with%
\begin{equation}
Sym\,F(s,t):=F(s,t)+F(t,s)+F(s,u)+F(u,s)+F(u,t)+F(t,u)  \label{int12}
\end{equation}%
is an amplitude, which satisfies Mandelstam analyticity -- as introduced in
\cite{Mandelstam:1958} -- and crossing symmetry. The amplitudes constructed
in the following sections have this form. But an extension to the isospin-1
crossing symmetry of charged pions is easily possible, as can be seen from
the cited literature.

The absorptive part in the $s$-channel is
\[
A_{s}(s,t)\equiv \mathrm{Abs}_{s}A(s,t):=(2i)^{-1}\left(
A(s+i0,t)-A(s-i0,t)\right) ,\quad \,s\geq 4;
\]%
it agrees with the imaginary part $\mathrm{Im}\,A(s+i0,t)$ if $-s<t<4$. The
amplitude is normalized such that the total cross section is given by%
\begin{equation}
\sigma _{total}(s)=\frac{8\pi }{\sqrt{s(s-4)}}\,A_{s}(s,0)=\frac{8\pi }{%
\sqrt{s(s-4)}}\,\mathrm{Im}\,A(s+i0,0)  \label{int4}
\end{equation}%
The partial wave expansion in the physical domain is%
\begin{equation}
A(s+i0,t)=2\sqrt{\frac{s}{s-4}}\sum_{l}(2l+1)a_{l}(s)P_{l}(z),\,s\geq 4
\label{int5}
\end{equation}%
with $z=1+\frac{2t}{s-4}$. The partial waves are given by%
\begin{equation}
a_{l}(s)=\frac{1}{2}\left[ s(s-4)\right] ^{-\frac{1}{2}%
}\int_{4-s}^{0}A(s+i0,t)\,P_{l}(z)dt,\,l=0,1,2,...  \label{int6}
\end{equation}%
For $l>n$ this integral is equivalent to the Froissart-Gribov integral%
\begin{equation}
a_{l}(s)=\frac{2}{\pi }\left[ s(s-4)\right] ^{-\frac{1}{2}}\int_{4}^{\infty
}A_{t}(s+i0,t)Q_{l}\left( 1+\frac{2t}{s-4}\right) dt.  \label{int7}
\end{equation}%
Due to the crossing symmetry (\ref{int3}) all odd partial wave amplitudes
vanish.

The amplitudes should satisfy at least partly the following unitarity
constraints, which can be formulated with the two-particle scattering
amplitude alone:

\begin{enumerate}
\item[4.] The elastic unitarity identities%
\begin{equation}
\mathrm{Im}\,a_{l}(s)=\left\vert a_{l}(s)\right\vert ^{2}\;\mathrm{for}%
\;l=0,1,2,...,\;\mathrm{and}\;4\leq s\leq 16.  \label{int8}
\end{equation}

\item[5.] The inelastic unitarity inequalities
\begin{equation}
1\geq \mathrm{Im}\,a_{l}(s)\geq \left\vert a_{l}(s)\right\vert ^{2}\;\mathrm{%
for}\;l=0,1,2,....  \label{int9}
\end{equation}%
and energies $s\geq 16$. In the sequel we also consider amplitudes, which
fulfill the inequalities (\ref{int9}) for all $s\geq 4$, but elastic
unitarity does not hold.
\end{enumerate}

The threshold behaviour of amplitudes, which satisfy the inequalities (\ref%
{int9}) but not elastic unitarity, can be rather arbitrary. The
constructions presented in this paper satisfy the following uniform
estimate, which is compatible with elastic unitarity:

\begin{enumerate}
\item[6.] The absorptive part in the $s$-channel has the bound%
\begin{equation}
\left\vert A_{s}(s,t)\right\vert \leq const\sqrt{\frac{s-4}{s}}\left(
1+\left\vert s\right\vert +\left\vert t\right\vert \right) ^{n}
\label{int10}
\end{equation}%
for $s\geq 4$ and $t\in \mathbb{C}_{cut}=\mathbb{C}\backslash \left( -\infty
,-s\right] \cup \left[ 4,\infty \right) $ including the boundary values at $%
t\pm i0$ if $t\geq 4$ or $t\leq -s$.
\end{enumerate}

Finally we give some mathematical notations used in this paper. A number $%
s\in \mathbb{R}$ is called \textit{positive} if $s\geq 0$ and \textit{%
strictly positive} if $s>0$. For a real variable $s$ and a complex parameter
$\lambda $ the function $\mathbb{R}\ni s\rightarrow s_{+}^{\lambda }\in
\mathbb{C}$ is defined as $s_{+}^{\lambda }=0$ if $s\leq 0$ and $%
s_{+}^{\lambda }=s^{\lambda }$ if $s>0$. The hat in $\hat{F}(s,t)$ indicates
that the function is symmetric with respect to its variables, $\hat{F}(s,t)=%
\hat{F}(t,s)$.

\section{Restrictions for the total cross section\label{bounds}}

In this section we shortly recapitulate bounds and constraints of the
absorptive part of the forward amplitude. From general principles follows
that the boundary values $A(s+i0,t),\,s\geq 4$, are integrable functions,
which have no meaning pointwise. Bounds like the Froissart bound are
therefore bounds for local averages, see e. g. section 17.1 of \cite%
{BLOT:1990}. But for amplitudes with uniformly continuous boundary values --
as considered in this paper -- the local statements can be used.

\subsection{The Froissart bound\label{F}}

The bound
\begin{equation}
\mathrm{Im}\,A(s+i0,0)\leq const\,s\left( \log s\right) ^{2},\;s\geq 4,
\label{b1}
\end{equation}%
has first been derived by Froissart \cite{Froissart:1961} assuming the
Mandelstam representation with a finite number of subtraction and bounded
partial waves. Then Martin \cite{Martin:1963,Martin:1966,Martin:1969} has
succeeded to derive this bound from the analyticity domain of axiomatic
quantum field theory and the unitarity constraints%
\begin{equation}
0\leq \mathrm{Im}\,a_{l}(s)\leq 1,\;l=0,1,2,...,\;s\geq 4.  \label{b2}
\end{equation}%
The Froissart bound is therefore a consequence of only a small part of the
analyticity and unitarity constraints of a two-particle scattering
amplitude. Especially, it is valid without assumptions about the crossed
channels. As already mentioned, it is possible to saturate this bound with
an amplitude, which satisfies the constraints 1. -- 3. and 5. of section \ref%
{not}. But elastic unitarity is missing in this construction; see section %
\ref{Froissart}.

\subsection{A bound for amplitudes with positive spectral functions\label%
{bpos}}

For the construction -- or for the mere proof of existence -- of amplitudes,
which satisfy elastic unitarity, analytic functions with positive spectral
functions turn out to be of exceptional importance; see the detailed
discussion in section \ref{dsf}. It is therefore of some interest to know
that the positivity of spectral functions leads to a more restrictive bound
of the amplitudes. The following bound has been derived by Goebel \cite%
{Goebel:1961} and by Martin \cite{Martin:1969a}: \newline
If $A(s,t)$ satisfies Mandelstam analyticity and has positive double
spectral functions, then the inelastic unitarity inequalities (\ref{b2})
imply the bound
\begin{equation}
0\leq \mathrm{Im}\,A(s+i0,0)\leq const\,s(\log s)^{-1},\;s\geq 4.  \label{b3}
\end{equation}%
Using additional consequences of the inequalities (\ref{int9}) Martin has
also deduced that these amplitudes satisfy a Mandelstam representation of at
most two subtractions.

Hence increasing total cross sections are excluded for unitary amplitudes
with positive double spectral functions. The bound (\ref{b3}) is an optimal
bound in the following respect. It can be saturated by crossing symmetric
amplitudes, which have positive double spectral functions and satisfy the
inelastic unitarity inequalities (\ref{int9}); see section 2.3 of \cite%
{Kupsch:1971}.

\subsection{Gribov's Theorem\label{G}}

In 1960 Gribov \cite{Gribov:1960,Gribov:1961} derived a consequence of
crossing symmetry and elastic unitarity without any assumption about the
partial waves at high energy. Let $A(s,t)$ be an amplitude, which satisfies
Mandelstam analyticity. Then elastic unitarity in the $t$-channel does not
allow a linear increase of the absorptive part like
\begin{equation}
A_{s}(s,t)\simeq s\,f(t)\quad \mathrm{if}\;s\rightarrow \infty  \label{b5}
\end{equation}%
for $t$ in some interval $t_{1}<t<t_{2}$ with $t_{1}<0$ and $t_{2}>4$.
Thereby $f(t)$ is a real analytic function with a cut starting at the
elastic threshold $t=4$. This statement is sometimes called \textit{Gribov's
Theorem}. With the same reasoning one can exclude an asymptotic behaviour $%
A_{s}(s,t)\simeq s^{p}\left( \log s\right) ^{q}\,f(t)$ if $s\rightarrow
\infty $ with any $p\in \mathbb{R}$ and $q\geq -1$, whereas values of $q<-1$
do not lead to a contradiction \cite{Froissart:1963}. Since the partial
waves remain bounded only if $p<1$ (and $q\in \mathbb{R}$) or $p=1$ and $%
q\leq 0$, Gribov conjectured that the total cross section has to decrease
for high energies.

In the language of Regge theory the result of Gribov means that the
asymptotics cannot be dominated by a fixed pole at real angular momentum.
One can circumvent the inconsistencies seen by Gribov in using rising
complex Regge trajectories, see section \ref{R-el}. But there remain
problems with the crossed channel contributions of the elastic unitarity
integral, if one tries to combine elastic unitarity with the inelastic
constraints (\ref{int9}); see section \ref{R-unitary}.

It should be stressed that Gribov's arguments are based on elastic unitarity
and crossing. If one demands crossing symmetry and the inelastic unitarity
inequalities (i. e. the constraints 1.-- 3. and 5. of section \ref{not}),
but not elastic unitarity, it is possible to find amplitudes which satisfy
these constraints and which have the high energy behaviour (\ref{b5}). The
construction of such amplitudes is shortly discussed in section \ref{Gribov}.

\section{Constructions using the Mandelstam representation\label{Mandelstam}}

In this section we recapitulate the construction of amplitudes which satisfy
crossing symmetry, elastic unitarity, and (partly) the inelastic unitarity
constraints (\ref{int9}) using the Mandelstam representation explicitly. For
these amplitudes elastic unitarity is obtained using a non-linear fixed
point mapping for the spectral functions. This method has first been
established by Atkinson for amplitudes which satisfy a Mandelstam
representation without subtraction \cite{Atkinson:1968a,Atkinson:1968b};
then it has been extended to amplitudes with one subtraction \cite%
{Atkinson:1970,Kupsch:1969b} and to amplitudes with an arbitrary (finite)
number of subtractions \cite{Atkinson:1969}.

The results for amplitudes with positive spectral functions are
recapitulated in section \ref{dsf}. The non-linear fixed point mapping for
elastic unitarity is presented in section \ref{fixedpoint}. The details are
given only for amplitudes without subtraction. Amplitudes with an arbitrary
number of subtractions are shortly discussed in section \ref{subt}.

\subsection{Amplitudes with positive spectral functions \label{dsf}}

Until now the proof of the existence of amplitudes, which satisfy all the
requirements of Mandelstam analyticity, crossing symmetry, exact elastic
unitarity, and the inelastic inequalities (\ref{int9}), has been given only
for amplitudes, which satisfy a Mandelstam representation with at most one
subtraction and which have positive spectral functions. The first proof has
been given by Atkinson for amplitudes, which satisfy an unsubtracted
Mandelstam representation $A(s,t)=Sym\,F(s,t)$ with
\begin{equation}
F(s,t)=\Phi _{0}\left[ \psi \right] (s,t):=\frac{1}{\pi ^{2}}%
\int_{4}^{\infty }\int_{16}^{\infty }\psi (x,y)\frac{1}{(x-s)(y-t)}dxdy,
\label{a1}
\end{equation}%
where $\psi (x,y)$ is a positive H\"{o}lder continuous function with support
in $\left[ 4,\infty \right) \times \left[ 16,\infty \right) $. The proof is
based on the construction of a non-linear fixed point equation for the
double spectral function such that any fixed point solution of this mapping
satisfies exactly elastic unitarity. The existence of such a fixed point is
established by the Leray-Schauder fixed point theorem \cite{Atkinson:1968a},
or by the contraction mapping theorem \cite{Atkinson:1968b}. The validity of
the inelastic constraints (\ref{int9}) follows from estimates of the partial
waves of this fixed point solution. In \cite{Atkinson:1968b} the
generalization to charged pions is considered, and in \cite{Atkinson:1969b}
solutions with additional CDD ambiguities are constructed. For all these
solutions the imaginary part of the forward amplitude vanishes like $\mathrm{%
Im}A(s+i0,0)\lesssim s^{\alpha }(\log s)^{-\beta }$ with $-1<\alpha <0$ and $%
\beta <-1$.

In the next step a class of amplitudes $A(s,t)=Sym\,F(s,t)$ has been
constructed which satisfy the once subtracted Mandelstam representation $%
F(s,t)=\Phi _{1}\left[ \varphi ,\psi \right] (s,t)$ with%
\begin{equation}
\Phi _{1}\left[ \varphi ,\psi \right] (s,t):=\frac{1}{2\pi }\int_{4}^{\infty
}\varphi (x)\frac{1}{x-s}dx+\frac{s+t}{\pi ^{2}}\int_{4}^{\infty
}\int_{16}^{\infty }\frac{\psi (x,y)}{x+y}\frac{1}{(x-s)(y-t)}dxdy
\label{a2}
\end{equation}%
where $\varphi (x)$ and $\psi (x,y)$ are positive spectral functions.
Elastic unitarity is again established by a fixed point equation for the
spectral functions. The inelastic constraints (\ref{int9}) are obtained by
explicit estimates for the partial wave amplitudes. The calculations for
these estimates are indicated in Appendix \ref{PW}. The imaginary part of
the forward amplitude can grow like%
\begin{equation}
\mathrm{Im}\,A(s+i0,0)\sim s^{\alpha }(\log s)^{-\delta }\quad \mathrm{if}%
\;s\rightarrow \infty  \label{a3}
\end{equation}%
with $0<\alpha <1$ and $\delta >1,$ or $\alpha =1$ and $\delta \geq 3$. The
amplitudes with $0<\alpha <1$ have been constructed in \cite{Kupsch:1969b}
for $\delta =2$; but solutions for any $\delta >1$ exist. The solutions with
$\alpha =1$ have been obtained by Atkinson \cite{Atkinson:1970}\footnote{%
The relation (1.1) in \cite{Atkinson:1970} should read $\sigma (s)\approx
(\log s)^{-3-\varepsilon }$ with $\varepsilon \geq 0$.}. The limitation to $%
\delta \geq 3$ is discussed below.

The non-linear mapping for elastic unitarity will be considered in more
detail in section \ref{fixedpoint}. Here we would like to add some comments
about the optimal increase of the amplitudes. Knowing the bound (\ref{b3})
and the arguments of Gribov \cite{Gribov:1960,Gribov:1961} one might expect
that solutions exist, which satisfy elastic unitarity together with the
inelastic bounds and increase like $\mathrm{Im}A(s+i0,0)\sim s(\log
s)^{-\delta }$ where $\delta $ can be a number just above $1$. The reason
for the restriction to $\delta \geq 3$ needs therefore some explanation. In
\cite{Atkinson:1970} the function $\psi (s,t)\geq 0$ is chosen such that%
\begin{equation}
\psi (s,t)\sim \frac{t}{s}(\log t)^{-\delta }(\log s)^{-\lambda }  \label{a6}
\end{equation}%
with$\,\delta >1$ and $\lambda >1$. This ansatz is motivated by Martin's
paper \cite{Martin:1969a} on amplitudes with positive double spectral
functions. Assuming a smooth behaviour of the spectral functions (H\"{o}lder
continuity) the absorptive part in the $t$-channel is bounded by $\left\vert
A_{t}(s,t)\right\vert \leq const\cdot t(\log t)^{-\delta }$ for $%
t\rightarrow \infty $, and the bound (\ref{a6}) is stable under the
iteration for elastic unitarity, see \cite{Atkinson:1970}. The imaginary
parts of the partial waves behave like $\mathrm{Im}\,a_{l}(s)\sim (\log
s)^{-\delta +1}$ and the real parts are bounded by $(\log s)^{-\delta +2}$
for high energies; see Appendix \ref{PW2}. With these estimates the
inelastic inequalities can be satisfied only if $\delta \geq 3$; and $%
\mathrm{Im}A(s,0)\sim s(\log s)^{-3}$ is the strongest increase at high
energies that can be derived with these constructions.

The weak part of the arguments leading to this conclusion is the estimate
for the real parts of the partial waves. If we only demand the inelastic
constraints (\ref{int9}) for $s\geq 4$, then crossing symmetric amplitudes
with positive double spectral function have been constructed with an
improved estimate of the partial waves, $\left\vert a_{l}(s)\right\vert
\lesssim c\cdot (\log s)^{-\delta +1}$, where any $\delta \geq 1$ is
admitted. Thereby the constraints 1. -- 3. and 5. of section \ref{not} can
be satisfied such that the bound (\ref{b3}) of the forward amplitude $%
\mathrm{Im}\,A(s+i0,0)\sim s(\log s)^{-1}$ is saturated, see section 2.3 of
Ref. \cite{Kupsch:1971}. To obtain the improved estimates for $\mathrm{Re}%
\,a_{l}(s)$ cancellations between the $s$-channel and the $u$-channel
contributions to $A_{t}(s,t)$ have to be taken into account. That is
possible with an explicit ansatz (as done in \cite{Kupsch:1971}). But if
elastic unitarity is incorporated, the norm estimates used for the fixed
point mapping are not so precise to extract these cancellations.
Nevertheless there remains a chance that the unitarity mapping can be
treated with modified norms to reach values $\delta <3$. But the restriction
to $\delta >1$ will remain as a consequence of the bound (\ref{b3}) and of
Gribov's arguments.

\begin{remark}
One can allow small negative contributions to the spectral functions. But
the estimates have to be dominated by the positive parts. Especially, all
spectral functions have to be positive near the boundary of their support.
For the double spectral functions that statement is true independent of any
specific construction: Mahoux and Martin \cite{Mahoux/Martin:1964} have
derived domains where the double spectral functions have to be positive
assuming elastic unitarity, crossing symmetry and the positivity of $\mathrm{%
Im}\,a_{l}(s)$ for $s\geq 4$ and $l=0,1,2,..$.
\end{remark}

\subsection{The fixed point mapping for elastic unitarity\label{fixedpoint}}

In the case of Mandelstam analyticity elastic unitarity implies the
following identity for the double spectral function \cite{Mandelstam:1958}%
\begin{equation}
\rho (s,t)=\int_{4}^{\infty }dt_{1}\int_{4}^{\infty
}dt_{2}~K(s,t,t_{1},t_{2})A_{t}(s+i0,t_{1})A_{t}(s-i0,t_{2}),  \label{el-4}
\end{equation}%
valid in the domain $4\leq s\leq 16$ and $t\geq 4$. The function $K$ is the
Mandelstam kernel%
\begin{equation}
K(s,t,t_{1},t_{2})=\frac{2}{\pi \sqrt{s(s-4)}}\left(
t^{2}+t_{1}^{2}+t_{2}^{2}-2(tt_{1}+tt_{2}+t_{1}t_{2})-4\frac{tt_{1}t_{2}}{s-4%
}\right) _{+}^{-\frac{1}{2}},  \label{el-5}
\end{equation}%
The support properties of this kernel imply that $\rho (s,t)$ has
non-vanishing contributions only for
\begin{equation}
s>4,\,t>16+\frac{64}{s-4}\quad \mathrm{or}\quad t>16,\,s>4+\frac{64}{t-16}\,.
\label{el-6}
\end{equation}%
For amplitudes (\ref{a1}) without subtraction the identity (\ref{el-4}) is
equivalent to elastic unitarity; in the case of $n$ subtractions one needs
also the identities (\ref{int8}) for the partial wave amplitudes with $0\leq
l<n$ to determine the single spectral functions.

To obtain a well defined fixed point problem the space of double spectral
functions is equipped with a Banach space topology. In most publications
spaces of functions $f(s,t)$, which are H\"{o}lder continuous in both
variables have been used \cite%
{Atkinson:1968a,Atkinson:1968b,Atkinson:1969b,Atkinson:1970,AFJK:1976,Kupsch:1969b}%
. But it is also possible to work with an integral norm \cite%
{Kupsch:1970b,Kupsch:1977}. This integral norm is more adequate for
amplitudes with Regge poles, see section \ref{R-el}.

The space $\mathcal{L}_{\gamma },\gamma >-\frac{1}{2}$, is the Hilbert space
of all complex functions $\mathbb{R}_{+}\ni t\rightarrow f(t)\in \mathbb{C}$
which have a finite norm%
\begin{equation}
\left\Vert f\right\Vert _{\gamma }=\left[ \int_{0}^{\infty }\left\vert
t^{-\gamma }f(t)\right\vert ^{2}\left( 1+\left\vert \log t\right\vert
^{2}\right) \frac{dt}{t}\right] ^{\frac{1}{2}}.  \label{el-7}
\end{equation}%
Functions $f(s,t)$ of two variables $s\geq 4$ and $t\geq 0$ are defined as H%
\"{o}lder continuous mappings $\left[ 4,\infty \right) \ni s\rightarrow
f(s,\,.\,)\in \mathcal{L}_{\gamma }$. More explicitly, a family of Banach
spaces $\mathcal{L}(\gamma ,\delta ),\,\gamma >-\frac{1}{2},\,\delta \in
\mathbb{R}$, is introduced with the norms%
\begin{equation}
\left\Vert f(s,t)\right\Vert _{\gamma ,\delta }=\sup_{s\geq 4,0<h\leq
1}s^{-\delta }\left( \left\Vert f(s,t)\right\Vert _{\gamma }+h^{-\mu
}\left\Vert f(s+h,t)-f(s,t)\right\Vert _{\gamma }\right) .  \label{el-8}
\end{equation}%
Thereby $\mu $ is a H\"{o}lder index from the interval $0<\mu <\frac{1}{2}$.
Let $f(s,t)$ be a function with support in $s\geq 4$ and $t\geq 4$, then the
function $\check{f}(s,t):=f(t,s)$ is obtained by an interchange of the
variables. For spectral functions $\rho \in \mathcal{L}(\gamma ,\delta )$
with $\gamma <0$ and $\delta <0$ the double dispersion integrals (\ref{a2}) $%
\Phi _{0}\left[ \rho \right] $ and $\Phi _{0}\left[ \check{\rho}\right] $
are well defined.

We consider the case without subtraction in some detail. Let $\psi (s,t)\in
\mathcal{L}(\gamma ,\gamma ),\newline
-\frac{1}{2}+\mu <\gamma <0$, be a real double spectral function, which has
a support in the (slightly extended) elastic domain $4\leq s\leq
16+\varepsilon ,\,\varepsilon >0$, and $16\leq t<\infty $, and let $\omega
(s,t)\in \mathcal{L}(\gamma ,\gamma )$ be a real double spectral function,
which has a support only in the inelastic region $s\geq 16$ and $t\geq 16$.
Then the total amplitude is defined as%
\begin{eqnarray}
&&A(s,t)=E(s,t)+B(s,t)\quad \mathrm{with}  \label{el-9} \\
&&E(s,t):=Sym\,\Phi _{0}\left[ \psi \right] (s,t)\quad \mathrm{and}\quad
B(s,t):=Sym\,\Phi _{0}\left[ \omega \right] (s,t).  \label{el-10}
\end{eqnarray}%
The double spectral function of $A$ is $\rho (s,t)=\psi (s,t)+\psi
(t,s)+\omega (s,t)+\omega (t,s)$. As a consequence of the support
restrictions for the spectral functions we have $\rho (s,t)=\psi (s,t)$ if $%
4\leq s\leq 16$. Hence, if the identity
\begin{equation}
\psi (s,t)=\lambda (s)\int_{4}^{\infty }dt_{1}\int_{4}^{\infty
}dt_{2}~K(s,t,t_{1},t_{2})A_{t}(s+i0,t_{1})A_{t}(s-i0,t_{2})  \label{el-11}
\end{equation}%
is true for $s\geq 4$, the amplitude satisfies elastic unitarity. Thereby $%
\lambda (s)$ is a differentiable function with the properties $0\leq \lambda
(s)\leq 1$ if $s\in \mathbb{R}$, and%
\begin{equation}
\lambda (s)=\left\{
\begin{array}{l}
1\;\mathrm{if}\;4\leq s\leq 16, \\
0\;\mathrm{if}\;s\geq 18.%
\end{array}%
\right.  \label{el-12}
\end{equation}%
This function cuts off the support of $\psi (s,t)$ at $s=18$.

The fixed point mapping is now defined for the double spectral function $%
\psi (s,t)$ as follows. We fix a H\"{o}lder index $\mu $ with $0<\mu <\frac{1%
}{2}$ and real parameters $\gamma $ and $\delta $ with the constraints $-%
\frac{1}{2}+\mu <\gamma <\delta <0$. Then a background contribution $B(s,t)$
is specified that has a double spectral function $\omega \in \mathcal{L}%
(\gamma ,\gamma )$ with support in the inelastic region $s\geq 16$ and $%
t\geq 16$. Choosing a double spectral function $\psi (s,t)\in \mathcal{L}%
(\gamma ,\gamma )$ with support in the domain $\left\{ (s,t)\mid 4\leq s\leq
18,\;t\geq 16\right\} $, the absorptive part $A_{t}(s,t)=\mathrm{Abs}%
_{t}A(s,t)$ of the amplitude (\ref{el-9}) is calculated as element of $%
\mathcal{L}(\gamma ,\delta )$%
\begin{eqnarray}
&&A_{t}(s,t)=D\left[ \psi \right] (s,t)+B_{t}(s,t)\quad \mathrm{with}
\label{el-13} \\
&&D\left[ \psi \right] (s,t):=\mathrm{Abs}_{t}E(s,t)=\frac{1}{\pi }%
\int_{4}^{\infty }\left( \psi (x,t)+\psi (t,x)\right) \left( \frac{1}{x-s}+%
\frac{1}{x-4+s+t}\right) dx.  \label{el-14}
\end{eqnarray}%
In (\ref{el-14}) the integral $\pi ^{-1}\int_{4}^{\infty }\psi
(x,t)(x-s)^{-1}dx$ is the $t$-channel absorptive part of the function $\Phi
_{0}\left[ \psi \right] (s,t)$; the other contributions arise from crossing.
Then the unitarity integral (\ref{el-4}) determines an image function
\begin{equation}
\psi ^{\prime }(s,t)=\lambda (s)\int_{4}^{\infty }dt_{1}\int_{4}^{\infty
}dt_{2}~K(s,t,t_{1},t_{2})A_{t}(s+i0,t_{1})A_{t}(s-i0,t_{2})  \label{el-15}
\end{equation}%
which is again an element of $\mathcal{L}(\gamma ,\gamma )$ with a support
in $4\leq s\leq 18$ and $16+64(s-4)^{-1}\leq t<\infty $. The equations (\ref%
{el-13}) -- (\ref{el-15}) define a non-linear mapping $\psi \rightarrow \psi
^{\prime }$ for the real spectral function $\psi \in \mathcal{L}(\gamma
,\gamma )$. This mapping is denoted as $\psi ^{\prime }=\Upsilon \left( \psi
\right) $. Any fixed point of this mapping yields an amplitude (\ref{el-9}),
which satisfies crossing symmetry and elastic unitarity. Moreover, the fixed
point solution is H\"{o}lder continuous in both variables \cite{Kupsch:1970b}
though we have used an integral norm. Therefore the assumption of H\"{o}lder
continuity -- as done in section \ref{not} -- is a natural one. A more
detailed analysis shows that $\Upsilon \left( \psi \right) $ is a
contraction mapping inside some ball $\left\Vert \psi \right\Vert \leq c$
provided the inhomogeneous contribution $\omega $ has a sufficiently small
norm. Hence a unique fixed point solution can be obtained by iteration. If $%
\omega \neq 0$ the trivial solution $\psi =0$ is excluded. Moreover, if $%
\omega (s,t)$ is a positive function with dominant contributions near the
boundary of its support, the estimates of Appendix \ref{PW2} can be used to
derive that the fixed point solution satisfies the inelastic unitarity
constraints (\ref{int9}) for $s\geq 16$.

\begin{remark}
The conditions for the Banach contraction principle (or for the
Leray-Schauder fixed point theorem) are only established for amplitudes with
a sufficiently small norm. Hence this method does not lead to amplitudes
with resonances. But a generalization to amplitudes with CDD ambiguities is
possible \cite{Atkinson:1969b}.
\end{remark}

\subsection{Amplitudes with an arbitrary number of subtractions\label{subt}}

The fixed point mapping $\Upsilon $ can be generalized to a mapping for
amplitudes with an arbitrary finite number of subtractions and for isospin-1
pions, see \cite{Atkinson:1969}. This mapping is defined on the Cartesian
products of the Banach spaces of the independent spectral functions, and it
is possible to satisfy the conditions for a contraction mapping. Hence
amplitudes exist, which satisfy the constraints 1.-- 4. of section \ref{not}
and have a polynomial increase for large $s$ and $t$ of arbitrary strength.
In the case of one subtraction the inelastic unitarity constraints can be
incorporated for all energies as already stated in section \ref{dsf}, but in
the case of more subtractions the partial waves are in general not bounded,
and the forward amplitude may increase polynomially.

\section{Regge type amplitudes that satisfy inelastic unitarity constraints
\label{Regge}}

In this section we review the construction of amplitudes that satisfy the
requirements 1.- 3. and 5. (together with 6.) of section \ref{not} -- only
elastic unitarity is missing. These amplitudes can produce a constant or
increasing total cross section. The construction starts from an explicit
ansatz with a Regge type asymptotics. The investigations include models with
simple and with double Regge poles \cite{Kupsch:1971} and models with
crossing Regge cuts \cite{Kupsch:1982} that saturate the Froissart bound.

In the case of a leading Regge pole at angular momentum $l=\alpha (t)$ the
asymptotic behaviour of the absorptive part is%
\begin{equation}
A_{s}(s,t)\simeq \beta (t)~s^{\alpha (t)}\;\mathrm{for}\;s\rightarrow \infty
\label{r1}
\end{equation}%
where $t$ lies in some interval, which includes $t=0$. The Regge trajectory $%
\alpha (t)$ and the residue function $\beta (t)$ are real analytic functions
with cuts starting at $t=4$. If the amplitude has a double pole at angular
momentum $l=\alpha (t)$ the amplitude has the asymptotic behaviour%
\begin{equation}
A_{s}(s,t)\simeq \beta (t)~s^{\alpha (t)}\log s\;\mathrm{for}\;s\rightarrow
\infty .  \label{r2}
\end{equation}%
The Froissart bound imposes the restriction $\alpha (0)\leq 1$ for the
intercept. In the subsequent section \ref{R-inel} we recapitulate the
construction of such Regge amplitudes. Thereby the limit case $\alpha (0)=1$
is included.

In \cite{Kupsch:1982} these methods have been extended to amplitudes with
crossing Regge cuts. The main step for the construction of such amplitudes
that saturate the Froissart bound is recapitulated in section \ref{Froissart}%
.

Since the double spectral functions of Regge amplitudes are oscillating and
increasing like $s^{\mathrm{Re}\,\alpha (t)}$ for $s\rightarrow \infty $,
the methods of section \ref{dsf} and of Appendix \ref{PW} are no longer
sufficient to derive the unitarity inequalities (\ref{int9}). As new
technique a linearization of the quadratic inequalities is used. This method
has been developed in \cite{Kupsch:1971,Kupsch:1982} and it is presented in
the next section before the details of Regge amplitudes are discussed.

\subsection{The linear unitarity inequalities\label{lin}}

Let $F(s,t)$ be an analytic function, which has partial wave amplitudes $%
f_{l}(s)$ that satisfy the relations%
\begin{eqnarray}
\left\vert f_{l}(s)\right\vert &\leq &c_{1},  \label{lin1a} \\
\left\vert \mathrm{Re}\,f_{l}(s)\right\vert &\leq &c_{2}\,\mathrm{Im}%
\,f_{l}(s)  \label{lin1b}
\end{eqnarray}%
for $l=0,1,2,...$ and $s\geq s_{1}\geq 4$ with some constants $c_{1}>0$ and $%
c_{2}\geq 0$. Then the inequalities $\left( \mathrm{Im}\,f_{l}\right)
^{2}\leq c_{1}\mathrm{Im}\,f_{l},$ and $\left( \mathrm{Re}\,f_{l}\right)
^{2}\leq c_{1}\left\vert \mathrm{Re}\,f_{l}(s)\right\vert \leq c_{1}c_{2}\,%
\mathrm{Im}\,f_{l}$ follow with the final result%
\begin{equation}
\left\vert f_{l}\right\vert ^{2}\leq c_{1}(1+c_{2})\,\mathrm{Im}\,f_{l},%
\hspace{1.5cm}l=0,1,2,..,s\geq s_{1}\geq 4.  \label{lin2}
\end{equation}%
After an appropriate scaling the partial wave amplitudes of $A(s,t)=c\cdot
F(s,t)$ with $0<c\leq \left( c_{1}+c_{1}c_{2}\right) ^{-1}$ satisfy the
normalized inelastic unitarity inequalities (\ref{int9}). Hence the
constraints (\ref{lin1a}) and (\ref{lin1b}) imply the quadratic unitarity
inequalities (\ref{int9}) for a scaled amplitude. The constraints (\ref%
{lin1a}) and (\ref{lin1b}) will be denoted as \textit{linear unitarity
inequalities}.

For the Regge amplitudes of this section the inelastic inequalities are
derived in the linear form (\ref{lin1a}) and (\ref{lin1b}). Thereby the
uniform bound (\ref{lin1a}) is obtained by a simple estimate of the integral
(\ref{int6}). The essential tool to derive the inequalities (\ref{lin1b}) is
an order relation between real functions \cite{Kupsch:1971,Kupsch:1982}.

We first define a set of functions $\phi (s,t)$ which have the following
properties:

\begin{enumerate}
\item[$\protect\alpha )$] The function $\phi (s,t)$ is H\"{o}lder continuous
in $s\geq s_{1}\geq 4$ and holomorphic in the cut plane $t\in \mathbb{C}%
_{cut}$, it is real if $s\geq s_{1}$ and $-s<t<4$.

\item[$\protect\beta )$] The partial wave amplitudes of $\phi (s,t)=2\sqrt{%
\frac{s}{s-4}}\sum_{l}(2l+1)f_{l}(s)P_{l}\left( 1+\frac{2t}{s-4}\right) $
are positive, $f_{l}(s)\geq 0,$ for $l=0,1,2,..,$ and $s\geq s_{1}$.
\end{enumerate}

The set of functions with these properties is denoted by $\mathcal{A}$, or
more precisely by $\mathcal{A}(s_{1})$. If $\phi _{1}(s,t)\in \mathcal{A}$
and $\phi _{2}(s,t)\in \mathcal{A}$, then the sum $\alpha \phi
_{1}(s,t)+\beta \phi _{2}(s,t)$ with $\alpha \geq 0,\,\beta \geq 0,$ and the
product $\phi _{1}(s,t)\phi _{2}(s,t)$ are also elements of $\mathcal{A}$.
As a consequence of these simple statements more complicated constructions
are possible; e.g., if $\phi (s,t)\in \mathcal{A}$, then also $\exp \lambda
\phi (s,t)\in \mathcal{A}$ if $\lambda \geq 0$. For more details see
Appendix \ref{PWpositive}. The statement $\phi (s,t)\in \mathcal{A}(4)$ is
the usual positivity constraint for the absorptive part of a scattering
amplitude \cite{Martin:1969}.

For functions which satisfy the conditions $\alpha )$ a partial order $\phi
_{1}\prec \phi _{2}$ can be defined by%
\begin{equation}
\phi _{1}(s,t)\prec \phi _{2}(s,t)\quad \mathrm{if}\quad \phi _{2}(s,t)-\phi
_{1}(s,t)\in \mathcal{A}.  \label{lin3}
\end{equation}%
This relation is preserved by multiplication with a function $\chi (s,t)\in
\mathcal{A}$%
\begin{equation}
\phi _{1}\prec \phi _{2}\quad \curvearrowright \quad \chi \cdot \phi
_{1}\prec \chi \cdot \phi _{2}.  \label{lin4}
\end{equation}%
The unitarity inequalities (\ref{lin1b}) are equivalent to the linear
relations%
\begin{equation}
-c_{2}\,\mathrm{Im}\,F(s+i0,t)\prec \mathrm{Re}\,F(s+i0,t)\prec c_{2}\,%
\mathrm{Im}\,F(s+i0,t)  \label{lin5}
\end{equation}%
with $s\geq s_{1}\geq 4$. To prove that the partial wave amplitudes of an
analytic function $F(s,t)$ satisfy the inequalities (\ref{lin2}) it is
therefore sufficient to derive the uniform bounds (\ref{lin1a}) and to
establish the order relations (\ref{lin5}) with a number $c_{2}\geq 0$.

The linear relations (\ref{lin1a}) and (\ref{lin1b}) are stronger than the
quadratic inequalities, and they are not compatible with the double spectral
region imposed by elastic unitarity.

\subsection{Amplitudes with Regge poles\label{R-inel}}

In this section we recapitulate the construction of crossing symmetric
amplitudes, which satisfy crossing symmetry and the inelastic unitarity
inequalities (\ref{int9}), and which exhibit the Regge asymptotics (\ref{r1}%
) or (\ref{r2}). All results have been derived for bounded Regge
trajectories $\alpha (t)$. Moreover the imaginary part is assumed to be
positive, $\mathrm{Im}\,\alpha (t+i0)\geq 0$, such that $\alpha (t)$ is a
convex increasing function for $t\in \left( -\infty ,4\right] $. The limit
case of a fixed pole is admitted. From general arguments we know that the
following restrictions must be obeyed: $\alpha (0)\leq 1$ (Froissart bound),
and $\alpha (4)<2$ \cite{Jin/Martin:1964b}. Amplitudes with the following
types of trajectories and residue functions have been constructed and will
be discussed in the sequel:

\begin{enumerate}
\item The leading Regge singularity is a moving pole at $\alpha (t)$. The
function $\alpha (t)$ is a bounded real analytic function with a cut
starting at $t=4$%
\begin{equation}
\alpha (t)=\alpha (\infty )+\frac{1}{\pi }\int_{4}^{\infty }\frac{\mathrm{Im}%
\,\alpha (x+i0)}{x-t}dx  \label{r5}
\end{equation}%
The imaginary part $\mathrm{Im}\,\alpha (t+i0)$ is H\"{o}lder continuous
with support in the interval $4\leq t<\infty $, it is strictly positive for $%
t>4$. The intercept $\alpha (0)$ has a value in the interval $0<\alpha
(0)\leq 1$, and the limit $\alpha (\infty )=\lim_{t\rightarrow \infty
}\alpha (t)$ can be any number $\alpha (\infty )<\alpha (0)$.

\item The leading Regge singularity is a fixed pole at $\alpha (t)=\alpha
_{0},\,t\in \mathbb{C}$, with $0<\alpha _{0}\leq 1$. Such amplitudes do not
contradict Gribov's Theorem since elastic unitarity is not satisfied. A
construction of the fixed pole solutions is indicated in section \ref{Gribov}%
. Martin and Richard \cite{Martin/Richard:2001} have presented another
simpler construction of an amplitude with a fixed pole at $\alpha
(t)=1,\,t\in \mathbb{C}$. But their proof of the inelastic inequalities does
not include all partial wave amplitudes.

\item The Regge singularity is a double pole on a rising trajectory $\alpha
(t)$ which has an intercept $\alpha (0)\leq 1$. The intercept can have the
value $\alpha (0)=1$ such that the total cross section (\ref{int4})
increases logarithmically. Moreover, there exist solutions with a fixed
double pole at $l=\alpha <1$.

\item The residue $\beta (t)$ is a real analytic function with a cut
starting at $t=4$. The partial wave coefficients of $\beta (t)$ are
positive, and the function $\beta (t)$ has the upper bound $\left\vert \beta
(t)\right\vert \leq const\,\left( 1+\left\vert t\right\vert \right)
^{-\delta },\,t\in \mathbb{C}_{cut}$, with $\delta >2^{-1}\left( 1+\alpha
(\infty )\right) $.
\end{enumerate}

\subsubsection{Amplitudes with rising Regge trajectories\label{risingR}}

The construction of amplitudes with Regge poles follows \cite{Kupsch:1971}.
The amplitudes have the structure
\begin{equation}
A(s,t)=\hat{R}(s,u)+\hat{R}(s,t)+\hat{R}(t,u)+G(s,t).  \label{r3}
\end{equation}%
Thereby $\hat{R}(s,u)$ is an $s-u$ symmetric Regge ansatz and $G(s,t)$ is an
$s-t-u$ crossing symmetric background function. The construction proceeds in
three steps:

\begin{enumerate}
\item In the first step an explicit ansatz for an $s-u$ symmetric Regge
function $\hat{R}(s,u)$ is given with the following properties:\newline
-- $\hat{R}(s,u)$ has the asymptotic behaviour (\ref{r1}),\newline
-- the crossed amplitudes $\hat{R}(s,t)$ and $\hat{R}(t,u)$ vanish for $%
s\rightarrow \infty $ and $t$ fixed.

\item The second step is the proof that $\hat{R}(s,u)$ satisfies (up to a
scale factor) the inequalities (\ref{int9}) for energies $s\geq 20$. For
that purpose the linear unitarity inequalities of section \ref{lin} are used.

\item The remaining task is the proof that one can find a crossing symmetric
amplitude $G(s,t)$, which satisfies a Mandelstam representation with
positive spectral functions (as discussed in section \ref{dsf}) such that (%
\ref{r3}) fulfills the unitarity inequalities (\ref{int9}) above threshold $%
s\geq 4$. If the trajectory function approaches negative values $\alpha
(\infty )<0$ for large momentum transfers, the background function $G(s,t)$
can be chosen to satisfy an unsubtracted Mandelstam representation.
\end{enumerate}

Thereby it is possible to keep the threshold behaviour (\ref{int10}) in all
steps of the construction.

To formulate an ansatz for $\hat{R}(s,u)$ a positive weight function $\sigma
(s)$ is introduced with support in the interval $17\leq s\leq 18$. This
weight function is normalized to $\int \sigma (s)ds=1$. The ansatz for $\hat{%
R}(s,u)$ is%
\begin{equation}
\hat{R}(s,u)=-\frac{\beta (t)}{\sin \pi \gamma (t)}\int_{17}^{18}ds^{\prime
}\sigma (s^{\prime })\int_{17}^{18}du^{\prime }\sigma (u^{\prime })\left[
(s^{\prime }-s)(u^{\prime }-u)\right] ^{\gamma (t)},  \label{r4}
\end{equation}%
where we have used $\gamma (t):=2^{-1}\alpha (t)$ to simplify the notations.
The absorptive part of $\hat{R}(s,u)$ behaves like (\ref{r1}) for $%
s\rightarrow \infty $ and $t$ fixed. The function $\sigma (s)$ should be
sufficiently differentiable such that the convolution with $\sigma (s)$
yields a smooth behaviour of $\hat{R}(s,u)$ at the thresholds of the
variables $s$ and $u$. The thresholds in $s$ and $u$ have been pushed to $%
s=17\;(u=17)$ such that (\ref{r4}) does not contribute to the double
spectral function in the elastic domain $s\leq 16$. But it is possible to
take any other threshold $s>4$ for this construction. For large real values
of $s$ and fixed $t$ the amplitude behaves like
\begin{equation}
\hat{R}(s+i0,u)\simeq -\frac{\beta (t)}{\sin \pi \gamma (t)}\exp \left(
-i\pi \gamma (t)\right) s^{\alpha (t)}\quad \mathrm{if}\;s\rightarrow \infty
.  \label{r4a}
\end{equation}%
It is possible to choose a residue function $\beta (t)$ such that the
crossed terms $\hat{R}(s,t)+\hat{R}(t,u)$ do not contribute to the
asymptotics for large $s$. If $\alpha (\infty )<0$ (and consequently $\gamma
(\infty )<0$) the factor $\left( \sin \pi \gamma (t)\right) ^{-1}$ has poles
in the physical region $t<0$. These poles have to be canceled by zeroes of
the residue function $\beta (t)$.

For a Regge trajectory with $\alpha (0)\leq 1$ and $\alpha ^{\prime }(0)\geq
0$ the integral (\ref{int6}) yields a uniform bound (\ref{lin1a}) for the
partial wave amplitudes%
\begin{equation}
\left\vert f_{l}(s)\right\vert \leq c\sqrt{\frac{s-4}{s}}s^{\alpha
(0)-1}\leq c\sqrt{\frac{s-4}{s}},\;l=0,1,2,...,\;s\geq 4.  \label{r6}
\end{equation}%
For energies $s\geq 18$ and $\,4-s\leq t\leq 0$ the imaginary and the real
part of (\ref{r4}) $\hat{R}(s+i0,u)$ are given by%
\begin{eqnarray}
&&\mathrm{Im}\,\hat{R}(s+i0,u)=\beta (t)\cdot \int ds^{\prime }\sigma
(s^{\prime })(s-s^{\prime })^{\gamma (t)}\int du^{\prime }\sigma (u^{\prime
})\left[ (u^{\prime }-4+s+t)\right] ^{\gamma (t)},  \label{r8} \\
&&\mathrm{Re}\,\hat{R}(s+i0,u)=-\beta (t)\cdot \cot \pi \gamma (t)\cdot \int
ds^{\prime }\sigma (s^{\prime })(s-s^{\prime })^{\gamma (t)}  \nonumber \\
&&\hspace{4cm}\times \int du^{\prime }\sigma (u^{\prime })(u^{\prime
}-4+s+t)^{\gamma (t)}.  \label{r9}
\end{eqnarray}%
Under appropriate assumptions about the residue function $\beta (t)$ it is
possible to derive the linear unitarity relations (\ref{lin5}) for these
functions if $s\geq 20$. Hence the Regge ansatz $\hat{R}(s,u)$ satisfies the
unitarity inequalities (\ref{lin2}) for $s\geq 20$. The extension to the
crossing symmetric amplitude (\ref{r3}) with the correct unitarity
inequalities (\ref{int9}) for $s\geq 4$ follows the constructions in section
3.3 of \cite{Kupsch:1971}. Some of the relevant calculations are given in
the Appendices \ref{unitRegge} and \ref{crossing}.

We would like to add three remarks.

\begin{remark}
\label{convex}To derive the inelastic unitarity inequalities for a sum like (%
\ref{r3}) the following statement is important. The set of amplitudes, which
satisfy the inelastic unitarity conditions (\ref{int9}), is a convex cone:
If $A_{1}(s,t)$ and $A_{2}(s,t)$ are two analytic functions with partial
wave amplitudes, which fulfil the relations (\ref{int9}), then the sum $%
\alpha _{1}A_{1}(s,t)+\alpha _{2}A_{2}(s,t)$ has the same property, if the
constants $\alpha _{1,2}$ are positive and $\alpha _{1}+\alpha _{2}\leq 1$.
\end{remark}

\begin{remark}
The Regge ansatz (\ref{r4}) has a rectangular support of its double spectral
region \newline
$\left\{ (s,t)\mid 17\leq s<\infty ,\,4\leq t<\infty \right\} $. Such a
support is compatible with the linear unitarity inequalities (\ref{lin1b}).
It is possible to choose the background term $G(s,t)$ such that the full
amplitude (\ref{r3}) has a double spectral region as required by elastic
unitarity.
\end{remark}

\begin{remark}
There are two essential ingredients for the construction of Regge amplitudes
presented in this section: The imaginary part of the trajectory function $%
\alpha (t)$ is positive, and the residue $\beta (t)$ function has positive
partial wave coefficients. These constraints are not compatible with elastic
unitarity, see section \ref{R-unitary}.
\end{remark}

\subsubsection{Amplitudes with a fixed pole\label{Gribov}}

The construction for amplitudes with a fixed pole at angular momentum $%
l=\alpha $ with $0<\alpha \leq 1$ is the same as for amplitudes with rising
trajectories. We start again from the ansatz (\ref{r4}). Since $\cot \pi
\gamma $ is a (finite) number the prove of the linear relations (\ref{lin5})
for $\hat{R}(s+i0,u)$ is simpler. The partial wave amplitudes have again the
uniform bound (\ref{r6}). In the particularly interesting case $\alpha
=1\;(\gamma =\frac{1}{2})$ the real part of the amplitude $\hat{R}(s+i0,u)$
vanishes for $s\geq 18$, and the relations (\ref{lin5}) follow immediately
for these energies. The extension to a crossing symmetric amplitude (\ref{r3}%
) $A(s,t)$ is done as in the case of rising Regge trajectories.

\subsubsection{Amplitudes with double poles\label{double}}

The extension to amplitudes with the asymptotic behaviour (\ref{r2}) has
been given in section 4. of \cite{Kupsch:1971}. In the language of Regge
poles the behaviour (\ref{r2}) corresponds to a double pole at angular
momentum $l=\alpha (t)$. In the first step the ansatz (\ref{r4}) is
generalized to an amplitude $\hat{R}_{1}(s,u)$ with a leading term $\sim
s^{\alpha (t)}\cdot \log s$. Let $\sigma (x)\geq 0,\,x\in \mathbb{R}$, be a
positive continuous function on the real line which has its support in the
interval $17\leq x\leq 18$ and which is normalized to $\int \sigma (x)dx=1$.
One can take the function $\sigma (x)$ used in (\ref{r4}). Then $L(s):=\int
\sigma (x)\log (x-s)dx$ defines a smooth version of the logarithm on the cut
plane $\mathbb{C}_{cut}$. The boundary values at $s=x\pm i0,\,x\geq 4$, are H%
\"{o}lder continuous with any index $0<\mu <1$. The $s-u$ symmetric function%
\begin{equation}
H(s,t)=L(s)+L(4-s-t)-L(t+1)  \label{r34}
\end{equation}%
is holomorphic in $(s,t)\in \mathbb{C}_{cut}^{2}$, and it has the upper
bounds $\left\vert \mathrm{Im}\,H(s,t)\right\vert \leq 3\pi $ and \newline
$\left\vert \mathrm{Re}\,H(s,t)\right\vert \leq 2\log (1+\left\vert
s\right\vert )+const$ for $(s,t)\in \mathbb{C}_{cut}^{2}$. If $s\geq 20$ and
$-s<t<4$ it can be written as%
\begin{equation}
H(s\pm i0,t)=h(s,t)\mp i\pi ,  \label{r36}
\end{equation}%
with a positive real part $h(s,t)=\int \sigma (x)\log \left[
(s-x)(x-4+s+t)(x-t-1)^{-1}\right] dx>0$. For large $s$ the real part has the
asymptotic behaviour $h(s,t)\simeq 2\log s$. Moreover $h(s,t)$ has positive
partial wave amplitudes, $h(s,t)\in \mathcal{A}(20)$, see Appendix \ref%
{PWpositive}.

The Regge ansatz (\ref{r4}) is now modified to%
\begin{equation}
\hat{R}_{1}(s,u)=\frac{1}{2}H(s,t)\cdot \hat{R}(s,u)  \label{r37}
\end{equation}%
For $s\rightarrow \infty $ along the real line the absorptive part behaves
like%
\begin{equation}
\mathrm{Abs}_{s}\hat{R}_{1}(s,u)\simeq \beta (t)~s^{\alpha (t)}\log s.
\label{r38}
\end{equation}%
The partial wave amplitudes of $\hat{R}_{1}(s,u)$ are uniformly bounded, if
the trajectory has the properties $\alpha (0)<1$ and $\alpha ^{\prime
}(0)\geq 0$ or $\alpha (0)=1$ and $\alpha ^{\prime }(0)>0$. Therefore
amplitudes with double poles can be constructed either if the pole position
lies on a rising trajectory $\alpha (t)$ with an intercept $\alpha (0)\leq 1$%
, or if the pole has a fixed position at $l=\alpha <1$.

The linear unitarity relations (\ref{lin5}) are derived for $\hat{R}%
_{1}(s,u) $ with the techniques presented in section \ref{lin}. The
extension to a crossing symmetric amplitude%
\begin{equation}
A(s,t)=\hat{R}_{1}(s,u)+\hat{R}_{1}(s,t)+\hat{R}_{1}(t,u)+G(s,t)  \label{r39}
\end{equation}%
is done as for the amplitudes (\ref{r3}). The crossing symmetric function $%
G(s,t)$ is again a background term, which satisfies a Mandelstam
representation without or with one subtraction. The asymptotic behaviour for
$s\rightarrow \infty $ of the full amplitude $A(s,t)$ is dominated by $\hat{R%
}_{1}(s,u)$, and (\ref{r38}) implies (\ref{r2}). If $\alpha (0)=1$ the
forward amplitude behaves like $A(s,0)\simeq \hat{R}_{1}(s,u=4-s)\simeq
i\beta (0)\,s\left( \log s-i\frac{\pi }{2}\right) $ for $s\rightarrow
+\infty $

\subsection{Amplitudes that saturate the Froissart bound \label{Froissart}}

From an investigation of the forward peak of the scattering amplitude one
knows that an amplitude, which saturates the Froissart bound and obeys the
inelastic constraints must have a $\sqrt{-t}\log s$ shrinking of the forward
peak \cite{AKM:1971}. A behaviour of this type can be achieved by a Regge
model with complex conjugate Regge trajectories (poles or cuts), which cross
at $t=0$, see e. g. \cite{Eden:1971,Oehme:1972}. But this literature does
not tell us, how crossing symmetry can be imposed \cite{Khuri:1976}.
Fortunately, it is possible to follow the steps used for the Regge
amplitudes in the last sections. In \cite{Kupsch:1982} a class of amplitudes
has been constructed that saturate the bound and that rigorously satisfy the
requirements 1.-- 3. and 5. of section \ref{not}. The essential part of this
paper is the ansatz of an $s-u$ symmetric amplitude $\hat{F}(s,u)$ that
satisfies \newline
-- the requirements 1. and 2. of section \ref{not},\newline
-- the unitarity inequalities (\ref{int9}) (up to a scaling factor) beyond
some energy $s\geq s_{1}\geq 4$,\newline
-- the saturation of the Froissart bound. \newline
The full crossing symmetric amplitude, which fulfills all constraints with
exception of elastic unitarity, can then be obtained following the methods
developed for the Regge amplitudes.

We indicate a few steps of the construction of $\hat{F}(s,u)$. Starting from
the $s-u$ symmetric function (\ref{r34}) $H(s,t)=H(4-s-t,t)$ we define the
function%
\begin{equation}
\Gamma (s,t,z)=-\frac{1}{\cos (\pi /2)z}\exp \left[ \frac{1+z}{2}H(s,t)%
\right]  \label{f1}
\end{equation}%
of the complex variables $s,t$ and $z$. This function is holomorphic in the
variables $(s,t)$ in the Mandelstam domain $\mathbb{C}_{cut}^{2}$, and it is
holomorphic in the variable $z$ in the open unit disc $\left\{ z\mid
\left\vert z\right\vert <1\right\} \subset \mathbb{C}$ of the complex plane.
The $s$-channel absorptive part is, see (\ref{r36}),%
\begin{equation}
\mathrm{Abs}_{s}\Gamma (s,t,z)=\left( \tan \frac{\pi }{2}z+i\right) \exp %
\left[ \frac{1+z}{2}h(s,t)\right] \quad \mathrm{if}\;s\geq 18.  \label{f2}
\end{equation}

Let $\gamma (\xi )$ be a real analytic function of the variable $\xi \in
\mathbb{C}\backslash \left[ 2,\infty \right) $, vanishing at $\xi
=0,\;\gamma (0)=0$, and with a H\"{o}lder continuous and positive imaginary
part $\mathrm{Im}\,\gamma (x+i0)\geq 0$ for $\,x\geq 2$. We assume a uniform
bound%
\begin{equation}
\left\vert \gamma (\xi )\right\vert \leq \delta <1\;\mathrm{for}\;\xi \in
\mathbb{C}\backslash \left[ 2,\infty \right) .  \label{f4}
\end{equation}%
Then the function%
\begin{equation}
\Phi (s,t;\gamma ):=\frac{1}{\sqrt{t}}\left[ \Gamma (s,t,\gamma (\sqrt{t}%
))-\Gamma (s,t,\gamma (-\sqrt{t}))\right]  \label{f5}
\end{equation}%
is an $s-u$ symmetric amplitude, which is analytic in the Mandelstam domain.
For $t<0$ it develops two complex conjugate Regge trajectories of simple
poles at angular momenta $l=1+\gamma (\pm i\sqrt{-t})$, which cross at $t=0$%
. Moreover, the function $\Phi (s+i0,t;\gamma )$ satisfies the linear
unitarity relations of section \ref{lin} beyond the energy $s=20$. The proof
for this statement is given in section 3 of \cite{Kupsch:1982}.

The real part of $\gamma $ along the imaginary axis $\mathrm{Re}\,\gamma
(i\tau )=\mathrm{Re}\,\gamma (-i\tau )$ is a monotonically decreasing
function of $\tau \in \left[ 0,\infty \right) $ with values from $\gamma
(0)=0$ to $\gamma (\infty )<0$. For values of $s>18$ and $t<0$ the imaginary
part of (\ref{f5}) is given by%
\[
\mathrm{Im}\,\Phi (s+i0,t;\gamma )=\frac{2}{\sqrt{-t}}\exp \frac{\left(
1+a(t)\right) h(s,t)}{2}\cdot \sin \left( \frac{b(t)}{2}h(s,t)\right)
\]%
with $a(t)=\mathrm{Re}\,\gamma (i\sqrt{-t})<0$ and $b(t)=\mathrm{Im}\,\gamma
(i\sqrt{-t})>0$. The asymptotic behaviour is%
\[
\mathrm{Im}\,\Phi (s+i0,t;\gamma )\simeq 2\beta _{0}(t)s^{1+a(t)}\frac{\sin %
\left[ b(t)\log s\right] }{\sqrt{-t}}
\]%
with $\beta _{0}(t)=2\exp \left[ -\frac{1}{2}L(t)\left( 1+a(t)\right) \right]
$. At $t=0$ we have%
\begin{equation}
\mathrm{Im}\,\Phi (s+i0,0;\gamma )\sim s\log s\quad \mathrm{if}%
\;s\rightarrow \infty  \label{f7}
\end{equation}%
as for the amplitude (\ref{r38}) with the Regge double pole. But the
amplitude $\Phi (s,t;\gamma )$ has -- in contrast to the amplitudes of
section \ref{R-inel} -- a $\sqrt{-t}\log s$ shrinking of the forward peak.
It satisfies the bound%
\begin{equation}
\left\vert \Phi (s+i0,t;\gamma )\right\vert \leq c\cdot \chi
_{1,1}(s,t)\quad \mathrm{if}\;s\geq 4\;\mathrm{and}\;4-s\leq t\leq 0
\label{f8}
\end{equation}%
with some positive constant $c$ and the function%
\begin{equation}
\chi _{m,n}(s,t)=\left\{
\begin{array}{ll}
s(\log s)^{m}\left( 1+\sqrt{-t}\log s\right) ^{-n} & \quad \mathrm{if}%
\;-1\leq t\leq 0, \\
s^{p}(-t)^{-(1+p)/2} & \quad \mathrm{if}\;t<-1,%
\end{array}%
\right.  \label{f9}
\end{equation}%
where $p$ is a parameter in the interval $0<p<1$, and $m$ and $n$ are
integers.

The amplitude $H(s,t)\Phi (s,t;\gamma )$ saturates the Froissart bound, but
its partial wave amplitudes are not bounded for $s\rightarrow \infty $. It
is necessary to improve the shrinking of the forward peak. For that purpose
we introduce a smooth positive weight function $\rho (\lambda )$ which
satisfies%
\begin{equation}
\begin{array}{l}
\rho (\lambda )\in \mathcal{C}^{3}(\mathbb{R}),\;\rho (\lambda )\geq
0,\;\int \rho (\lambda )d\lambda =1, \\
\rho (\lambda )=0,\;\mathrm{if}\;\lambda \notin \left[ \lambda _{1},\lambda
_{2}\right] ,\;0<\lambda _{1}<\lambda _{2}<\delta ^{-1}.%
\end{array}%
\end{equation}%
If $\lambda \in \left[ \lambda _{1},\lambda _{2}\right] $ then $\lambda
\gamma (\xi )$ has the same properties as demanded for $\gamma (\xi )$, the
bound (\ref{f4}) is replaced by $\left\vert \lambda \gamma (\xi )\right\vert
\leq \lambda _{2}\delta <1$. The integral%
\begin{equation}
\widetilde{\Phi }(s,t)=\int_{\lambda _{1}}^{\lambda _{2}}\rho (\lambda )\Phi
(s,t;\lambda \gamma )d\lambda  \label{f11}
\end{equation}%
defines again an analytic function in the Mandelstam domain. The $\lambda $%
-integration leads to a rapidly decreasing function of $\sqrt{-t}\log s$
yielding a stronger shrinking of the forward peak%
\begin{equation}
\left\vert \widetilde{\Phi }(s+i0,t)\right\vert \leq const\cdot \chi
_{1,3}(s,t)\quad \mathrm{if}\;s\geq 4\;\mathrm{and}\;4-s\leq t\leq 0.
\label{f12}
\end{equation}%
whereas the forward amplitude $\widetilde{\Phi }(s,0)=\Phi (s,0;\gamma )$
remains unchanged.

Let $\beta (t)$ be a bounded real analytic function of the complex variable $%
t\in \mathbb{C}\backslash \left[ 4,\infty \right) $ with positive partial
wave coefficients. Then the $s-u$ symmetric amplitude%
\begin{equation}
\hat{F}(s,u)=\beta (t)H(s,t)\widetilde{\Phi }(s,t)  \label{f13}
\end{equation}%
saturates the Froissart bound. Since $\Phi (s+i0,t;\gamma )$ fulfils the
linear unitarity relations (\ref{lin5}), it is straightforward to derive
such relations also for $\hat{F}(s,u)$ beyond the energy $s\geq 20$. The
extension to a crossing symmetric amplitude
\[
A(s,t)=\hat{F}(s,u)+\hat{F}(s,t)+\hat{F}(t,u)+G(s,t)
\]%
which satisfies the inelastic constraints for all energies $s\geq 4$ follows
the same steps as used for the Regge amplitudes of section \ref{R-inel}. The
background function $G(s,t)$ can again be represented by a crossing
symmetric Mandelstam integral without subtraction and with a positive double
spectral function. For the details see \cite{Kupsch:1982}.

\section{Regge amplitudes that satisfy elastic unitarity\label{R-el}}

In the first part of this section we discuss the construction of amplitudes,
which satisfy\newline
-- Mandelstam analyticity and crossing symmetry in $s,t$ and $u$,\newline
-- elastic unitarity for $4\leq s\leq 16,$\newline
-- the inelastic unitarity inequalities $\left\vert a_{l}(s)\right\vert \leq
1$ for $s\geq 16$, \newline
-- Regge asymptotics with a trajectory $\alpha (t)$, which has an intercept $%
0<\alpha (0)\leq 1$.

\noindent If one does not care about unitarity constraints in the inelastic
region, the trajectory function $\alpha (t)$ can be rather arbitrary, and an
intercept $\alpha (0)>1$ is possible. The remaining -- and still unsolved
problem -- to prove the existence of a Regge amplitude, which satisfies
crossing symmetry, elastic unitarity and all inelastic constraints, is
discussed in section \ref{R-unitary}.

So far there is no extension of the fixed point equations of section \ref%
{R-fixedpoint} to amplitudes with double poles or to amplitudes with
crossing Regge trajectories and cuts as needed for the saturation of the
Froissart bound.

\subsection{The fixed point problem for Regge amplitudes\label{R-fixedpoint}}

It is possible to apply the fixed point method of section \ref{fixedpoint}
to amplitudes, which have a Regge ansatz of the type presented in section %
\ref{risingR} as inhomogeneous term, see e.g. \cite{Kupsch:1970a}\footnote{%
The Regge amplitudes constructed in section 2.1 of \cite{Kupsch:1970a} do
not have all the properties indicated there. But one can substitute these
erroneous Regge amplitudes by the amplitudes defined in section 3. of \cite%
{Kupsch:1971} (and in section \ref{risingR} of the present paper).}. But
then one has to assume that $\mathrm{Re}\,\alpha (t)<1$ for $t\leq 16$, and
in the physical region the asymptotics is dominated by the background
generated by the iteration. One therefore needs a more elaborate technique,
which exposes the Regge pole explicitly. For that purpose three approaches
have been used in the literature: \newline
-- the use of partial wave amplitudes and the Watson-Sommerfeld transform,
\newline
-- the use of partial wave amplitudes and N/D equations, and \newline
-- the Khuri representation of the amplitude.

Elastic unitarity can be easily formulated for partial wave amplitudes with
complex angular momenta. But the analyticity properties of the
Watson-Sommerfeld transform cause some problems in constructing a crossing
symmetric scattering amplitude with the correct Mandelstam analyticity.
These problems have been solved by Atkinson and his collaborators \cite%
{AFJK:1976,Frederiksen:1975}. But unfortunately the fixed point equation
obtained in \cite{AFJK:1976} does not guarantee elastic unitarity for the
pole term.

The N/D equations used by Johnson and Warnock \cite%
{Warnock2:1977,Warnock3:1977,Warnock4:1981} have the advantage that the
Regge pole appears as zero of the denominator function, and analyticity and
unitarity of the pole term is naturally included in the equations. But these
authors have not given a conclusive proof, whether and under what conditions
a fixed point solutions exists.

The following account is based on \cite{Kupsch:1977} and uses the Khuri
representation \cite{Khuri:1963}, in which the correct analyticity is easily
exposed, but the unitarity identity is more involved. The calculations can
also be performed with the Watson-Sommerfeld integral representation
(borrowing methods from the publications \cite{AFJK:1976,Frederiksen:1975}
to obtain the correct analyticity). Since elastic unitarity imposes a
non-linear relation between the Regge term and the background term, the
fixed point mapping has to refer to the degrees of freedom of the Regge
trajectory and of the background amplitude.

In the subsequent arguments the Regge poles appear in the large $t$
asymptotics. The trajectory is therefore written as function of $s$. We
consider an amplitude with one Regge trajectory $\alpha (s)$, which has the
following properties:

\begin{enumerate}
\item[a)] The function $\alpha (s)$ is real analytic and bounded. It has a H%
\"{o}lder continuous and strictly positive imaginary part $\mathrm{Im}%
\,\alpha (s+i0)>0$ if $s>4$.

\item[b)] The threshold behaviour of the trajectory is $\mathrm{Im}\,\alpha
(s+i0)\simeq c(s-4)_{+}^{\sigma +\frac{1}{2}},\,s\approx 4$, with the
exponent $\sigma =\alpha (4)$.

\item[c)] The values of $\alpha (s)$ are restricted to $-1<\alpha (\infty )<-%
\frac{1}{2},\;0<\alpha (4)<2$, and $\mathrm{Re}\,\alpha (s)<\gamma _{1}$ if $%
s\in \mathbb{C}_{cut}$. Thereby $\gamma _{1}$ is a constant in the open
interval $2<\gamma _{1}<3$.
\end{enumerate}

The strictly positive imaginary part of $\alpha (s+i0)$ in the elastic
interval $4<s\leq 16$ is needed to circumvent the problem encountered by
Gribov, see section \ref{G}. For the actual construction we assume the
stronger constraints a). The constraint b) is a consequence of elastic
unitarity \cite{Barut/Zwanziger:1962,Cheng/Sharp:1963}. The bounds c)
simplify the construction. Values of the intercept $\alpha (0)\approx 1$ are
included. Since we do not yet demand the inelastic constraints, also
solutions with $\alpha (0)>1$ exist.

The trajectory function is represented by%
\begin{equation}
\alpha (s)=\alpha \left[ \chi \right] (s)=\alpha _{0}(s)+\frac{s-4}{\pi }%
\int_{4}^{18}\sqrt{\frac{x-4}{x}}(x-4)^{\sigma -1}\frac{\chi (x)}{x-s}dx
\label{rel2}
\end{equation}%
The real analytic function $\alpha _{0}(s)$ is an input function for the
calculations. It has a cut at $s\geq 16$ and a positive imaginary part.
Since the value of $\alpha (4)$ determines the threshold behaviour of the
trajectory, we choose the subtraction point in (\ref{rel2}) at $s=4$. The
value of $\alpha _{0}(4)=\alpha (4)=\sigma $ is restricted to $0<\sigma <2$,
and we assume $-1<\alpha _{0}(\infty )<-\frac{1}{2}$. The trajectory (\ref%
{rel2}) is a functional of the strictly positive H\"{o}lder continuous
function $\chi (x)$, which is defined on the interval $4\leq x\leq 16$ and
extended to $x\geq 16$ by $\chi (x)=\chi (16)\lambda (x)$ with the cut off
function (\ref{el-12}).

The scattering amplitude has the structure (\ref{int12})
\begin{equation}
A(s,t)=Sym\,F(s,t)\quad \mathrm{with}\quad F(s,t)=R(s,t)+G(s,t).
\label{rel3}
\end{equation}%
The functions $R(s,t)$ and $G(s,t)$ are holomorphic in the domain (\ref%
{int11}). The term $R(s,t)$ has a Regge asymptotics $\mathrm{Abs}%
_{t}R(s,t)\simeq \beta (s)t^{\alpha (s)}$ for $t\rightarrow \infty $. The
function $G(s,t)$ is a background term. The ansatz (\ref{rel3}) corresponds
to a strip approximation \cite{Chew/Jones:1964}. To treat the Regge poles we
use the Khuri representation with respect to the variable $t$. The function $%
F(s,t)$ is represented by the Mellin-Barnes integral%
\begin{equation}
F(s,t)=-\frac{1}{2i}\int_{\mathcal{C}}\frac{f(s,\nu )}{\sin \pi \nu }%
(-t)^{\nu }d\nu .  \label{rel4}
\end{equation}%
The curve of integration $\mathcal{C}$ goes from $v=\gamma _{0}-i\infty $ to
$v=\gamma _{0}+i\infty $ with $-\frac{1}{2}<\gamma _{0}<0$, the (Khuri)
poles of the meromorphic function $f(s,\nu )$ lie to the left side of the
curve and the integers $\nu =0,1,2,...$ lie to the right side. The function $%
f(s,\nu )$ is the Mellin transform of $F_{t}(s,t)=\mathrm{Abs}_{t}F(s,t)$,
and the absorptive part of $F(s,t)$ in the $t$-channel is given by, see
Appendix \ref{Mellin},%
\begin{equation}
\mathrm{Abs}_{t}F(s,t)=\frac{1}{2i}\int_{\gamma _{1}}f(s,\nu )t^{\nu }d\nu .
\label{rel5}
\end{equation}%
The symbol \ $\int_{\gamma }d\nu ...$ means integration along the line $%
\left\{ \nu =\gamma +ix\mid -\infty <x<\infty \right\} $ with the fixed real
part $\mathrm{Re}\,\nu =$ $\gamma $. The line of integration in (\ref{rel5})
has to satisfy $\gamma _{1}>\sup \left\{ \mathrm{Re}\,\alpha (s)\mid s\in
\mathbb{C}_{cut}\right\} $. In the case without Regge poles as considered in
section \ref{fixedpoint} we have $F(s+i0,t+i0)\newline
=\Phi _{0}\left[ \psi \right] (s+i0,t+i0)\in \mathcal{L}(\gamma _{0},\delta
) $ with $-\frac{1}{2}+\mu <\gamma _{0}<\delta <0$, and the function $%
f(s,\nu ) $ is holomorphic in $\nu $ in the half plane $\left\{ \nu \mid
\mathrm{Re}\,\nu >\gamma _{0}\right\} $. Then the line of integration in (%
\ref{rel4}) and (\ref{rel5}) can be pushed back to $\mathrm{Re}\,\nu =\gamma
_{0}$. If a Regge pole at position $\nu =\alpha (s)$ enters the half plane $%
\left\{ \nu \mid \mathrm{Re}\,\nu >\gamma _{0}\right\} $ the function $%
f(s,\nu )$ has a series of Khuri poles at positions $\nu =\alpha (s),\alpha
(s)-1,\alpha (s)-2,...$. In a simplified picture, which neglects the Khuri
daughter poles (and the H\"{o}lder continuity in the variable $t$), the
separation of $f(s,\nu )$ into Regge pole contribution and holomorphic
background can be written as
\begin{equation}
f(s,\nu )=\frac{\beta (s)}{v-\alpha (s)}t_{1}^{\alpha (s)-\nu }+g(s,\nu ).
\label{rel7}
\end{equation}%
A more adequate pole term is given in Appendix \ref{pole}. The Mellin
transform $a(s,\nu )\newline
=\mathcal{M}\left[ A_{t}(s,t)\right] (\nu )$ of the full amplitude (\ref%
{rel3}) is then%
\begin{equation}
a(s,\nu )=f(s,\nu )+\mathrm{crossed\,terms}=\frac{\beta (s)}{v-\alpha (s)}%
t_{1}^{\alpha (s)-\nu }+b(s,\nu ).  \label{rel8}
\end{equation}%
Thereby the crossed Khuri pole terms contribute only to the holomorphic
background $b(s,\nu )$.

The elastic unitarity equation (\ref{el-11}) has been calculated for the
Khuri representation in \cite{Kupsch:1977}; a few details are also given
here in Appendix \ref{unitInt}. In this section we only use the truncated
form%
\begin{equation}
w^{\prime }(s,\nu )=\frac{1}{2}\lambda (s)\sqrt{\frac{s-4}{s}}(s-4)_{+}^{\nu
}B(1+\nu ,1+\nu )a(s+i0,\nu )a(s-i0,\nu ),  \label{rel9}
\end{equation}%
which exhibits the essential consequences of the exact identity (\ref{ru5}).
Thereby $B(x,y)$ is the Euler beta function. The function $w^{\prime }(s,\nu
)=\mathrm{Abs}_{s}f^{\prime }(s,\nu )$ is the Mellin transform of the double
spectral function of the image amplitude $F^{\prime }(s,t)$ generated by the
fixed point mapping. Without Khuri poles we can use the dispersion integral
\begin{equation}
f^{\prime }(s,\nu )=\frac{1}{\pi }\int_{4}^{18}w^{\prime }(x,\nu )\frac{1}{%
x-s}dx  \label{rel10}
\end{equation}%
to calculate the Mellin transform of $\mathrm{Abs}_{t}F^{\prime }(s,t)$. But
if $a(s,\nu )$ has a pole contribution (\ref{rel8}), then in the general
case the integral (\ref{rel10}) does not produce a pole but a cut. One has
to satisfy rather exceptional conditions to keep a Regge pole stable under
the iteration. To see this condition we assume that $w^{\prime }(s,\nu
)=w(s,\nu )=(2i)^{-1}\left( f(s+i0,\nu )-f(s-i0,\nu )\right) $ is a solution
of the equation (\ref{rel9}). Comparing the residues at $\nu =\alpha (s+i0)$
on both sides of (\ref{rel9}) we obtain the following identity in the
elastic interval $4\leq s\leq 16$
\begin{equation}
\overline{\beta }=\sqrt{\frac{s}{s-4}}(s-4)_{+}^{-\alpha }B^{-1}(1+\alpha
,1+\alpha )\,\mathrm{Im}\,\alpha -2i\,b(s-i0,\alpha )\,\mathrm{Im}\,\alpha
\label{rel11}
\end{equation}%
with $\alpha =\alpha (s+i0)$ and $\overline{\beta }=\beta (s-i0)=\overline{%
\beta (s+i0)}$. This identity relates the analytic functions $\alpha
(s),\,\beta (s)$ and the background term $b(s,\nu )$ in a highly non-trivial
manner. The ansatz (\ref{rel2}) with $\mathrm{Im}\,\alpha (s+i0)\sim
(s-4)^{\sigma +\frac{1}{2}}$ compensates exactly the singularity of the
right hand side at the threshold $s=4$. An essential point is: The exact
form of the equation (\ref{rel9}) leads also to an identity, which can be
written in the form (\ref{rel11}); only the interpretation of the
holomorphic function $b(s,\nu )$ has changed. The same type of identity
emerges, if one performs the calculations with partial wave amplitudes.

The fixed point mapping for elastic unitarity has to take into account the
constraint (\ref{rel11}). In \cite{Kupsch:1977} a non-linear mapping $\chi
\rightarrow \chi ^{\prime }=T\left[ \chi ,b\right] $ for the unknown
function $\chi (x)$ in (\ref{rel2}) was built up from the real and the
imaginary parts of (\ref{rel11}) and the Hilbert transforms between the
imaginary and the real parts of $\alpha (s)$ and of $\beta (s)$. This
mapping preserves the positivity of $\mathrm{Im}\,\alpha (s+i0)$: if $\chi
(s)$ is positive for $4\leq s\leq 16$, the image $\chi ^{\prime }(s)$ is
also a positive function. Given the background function $b(s,\nu )$ the
fixed point solution $\chi =T\left[ \chi ,b\right] $ leads to analytic
functions $\alpha (s)$ and $\beta (s)$, which satisfy the identity (\ref%
{rel11}). Thereby a ghost killing factor $\alpha (s)$ can be inserted into
the residue function $\beta (s)$ in order to cancel the pole of $%
R(s,t)+R(s,u)$ at angular momentum $l=0$, see Appendix \ref{pole}. Moreover
one can prescribe a strong decrease of $\beta (s)$ such that the crossed
Regge terms contribute only to the background.

The background is determined by a mapping $b\rightarrow b^{\prime }=\Upsilon %
\left[ \chi ,b\right] $, which incorporates elastic unitarity and crossing
for the background if the Regge term is given. The full non-linear mapping
is the Cartesian product $T\times \Upsilon $
\begin{equation}
\left(
\begin{array}{c}
\chi \\
b%
\end{array}%
\right) \rightarrow \left(
\begin{array}{c}
\chi ^{\prime } \\
b^{\prime }%
\end{array}%
\right) =\left(
\begin{array}{c}
T\left[ \chi ,b\right] \\
\Upsilon \left[ \chi ,b\right]%
\end{array}%
\right) .  \label{rel12}
\end{equation}%
A detailed norm estimate of all steps yields that the square of this mapping
is a contraction provided some norm restrictions are satisfied. The fixed
point solution can therefore be obtained by iteration.

In \cite{Kupsch:1977} the additive function $G(s,t)$ in (\ref{rel3}) is
omitted, and the role of the inhomogeneous term is taken over by an additive
(sufficiently small) constant $\tau >0$ in the dispersion relation for $%
\beta (s)$.The fixed point solution depends continuously on this constant.
The norm restrictions imply that $\chi $ and $\beta $ are small, but they
are strictly positive (at least for $4\leq s\leq 16$) if $\tau >0$. The
shape of the trajectory $\alpha (s)$ is essentially determined by the
function $\alpha _{0}(s)$, which can be chosen within the constraints given
above. The fixed point solution of (\ref{rel12}) depends continuously on $%
\alpha _{0}(s)$. The intercept of the trajectory is%
\begin{equation}
\alpha (0)=\alpha _{0}(0)-\frac{4}{\pi }\int_{4}^{18}x^{-\frac{3}{2}%
}(x-4)^{\sigma -\frac{1}{2}}\chi (x)dx<\alpha _{0}(0).  \label{rel13}
\end{equation}%
If we start with an input function $\alpha _{0}(s)$, which has an intercept $%
\alpha _{0}(0)=1$, the intercept of $\alpha $ is smaller, $\alpha (0)<1$.
But we can start with a function $\alpha _{0}(s)$ which has an intercept $%
\alpha _{0}(0)=1+\delta ,\,0<\delta <1$. Choosing a small parameter $\tau >0$
fixed point solutions with an arbitrarily small $\sup_{s}\chi (s)$ exist.
Hence there are solutions $\chi (s)$ such that $\alpha (0)>1$. The existence
of amplitudes (\ref{rel3}) with a Regge trajectory $\alpha (s)$ which has
exactly the intercept $\alpha (0)=1$ is then the consequence of the
continuity of the fixed point solution with respect to the input function $%
\alpha _{0}(s)$.

If $\alpha (0)\leq 1$ the partial wave amplitudes of $R(s,t)+R(s,u)$ are
uniformly bounded. One can achieve that the residue function $\beta (s)$
decreases fast enough such that crossed terms do not contribute to the
leading asymptotic behaviour, and the partial wave amplitudes of the full
amplitude (\ref{rel3}) are also uniformly bounded. This bound is a
continuous function of the norm of the amplitude. Since we can obtain
solutions with arbitrarily small norm, solutions with partial wave
amplitudes bounded by unity exist.

\subsection{Both constraints: elastic and inelastic unitarity\label%
{R-unitary}}

The inelastic unitarity inequalities (\ref{int9}) include the weaker
constraint of positivity:

\begin{itemize}
\item The amplitude $A(s,t)$ satisfies \textit{positivity} if $\mathrm{Im}%
\,a_{l}(s)\geq 0$ for $l=0,1,2,...$ and all energies $s\geq 4$.
\end{itemize}

Up to now there is no conclusive argument that the fixed point equations of
section \ref{R-fixedpoint} have solutions, which satisfy positivity. A
fortiori there is no proof that the fixed point equations have solutions,
which satisfy the full inelastic unitarity inequalities.

We give some details for a better understanding of this remaining problem.
From Mahoux and Martin \cite{Mahoux/Martin:1964} we know a consequence of
elastic unitarity: the double spectral function has to be positive in a
(precisely defined) neighbourhood of the boundary of its support. This fact
implies that the imaginary parts $\mathrm{Im}\,a_{l}(s)$ are positive for
sufficiently large angular momenta, $l\geq \Lambda (s)$, and energies $s>20$%
. Thereby $\Lambda (s)$ is an unknown function with the range $0\leq \Lambda
(s)<\infty $.

One can add an inhomogeneous term $G(s,t)$ to the Regge term $R(s,t)$, see (%
\ref{rel3}). This term should be given by an unsubtracted crossing symmetric
Mandelstam integral with positive double spectral function. Borrowing some
arguments from Appendix \ref{crossing} it is possible to obtain solutions,
which satisfy crossing symmetry, elastic unitarity and positivity within
some finite range of energy.

But the proof of the positivity for the partial wave amplitudes with $%
l<\Lambda (s)$ at higher energies remains open. The problem originates from
the partial wave amplitudes of the crossed Regge term $R(t,s)$, which
dominates the $s$-channel asymptotics, $\mathrm{Abs}_{s}R(t,s)\simeq \beta
(t)s^{\alpha (t)}$ if $s\rightarrow \infty $. Unfortunately one cannot use
arguments from section \ref{R-inel}. There the inelastic inequalities (\ref%
{int9}) have been derived under the condition that $\mathrm{Im}\,\alpha
(t+i0)\geq 0$ and that $\beta (t)$ has positive partial wave coefficients.
But this property does not hold for solutions of equation (\ref{rel11}). To
see that we first derive

\noindent \textbf{Lemma 1}\ Let $\alpha (t)$ and $\beta (t)$ be real
analytic functions which satisfy the identity (\ref{rel11}). If $\mathrm{Im}%
\,\alpha (t+i0)=c_{1}(t-4)^{\sigma +\frac{1}{2}}+\mathcal{O}\left(
(t-4)^{\sigma +\frac{1}{2}+\mu }\right) $ for $t\rightarrow 4$ with $%
c_{1}>0,\,\sigma =\alpha (4)>0$ and $\mu >0$, then near threshold the
imaginary part of $\beta (t+i0)$ has the form%
\begin{equation}
\mathrm{Im}\,\beta (t+i0)=c_{2}\,(t-4)^{\sigma +\frac{1}{2}}\log (t-4)+%
\mathcal{O}\left( (t-4)^{\sigma +\frac{1}{2}}\right) ,  \label{rel14}
\end{equation}%
where $c_{2}$ is a positive constant.

\textbf{Proof}\ The functions $\mathbb{R}\ni t\rightarrow \alpha (t+i0)$ and
$\mathbb{R}\ni t\rightarrow b(t-i0,\alpha (t+i0))$ are H\"{o}lder continuous
with index $\mu \in \left( 0,\frac{1}{2}\right) $. The relation (\ref{rel11}%
) then implies
\begin{equation}
\mathrm{Im}\,\beta (t+i0)=c_{1}\,(t-4)^{\sigma -\mathrm{Re}\alpha
(t+i0)}\sin \left( \log (t-4)\cdot \mathrm{Im}\,\alpha \right) +\mathcal{O}%
\left( (t-4)^{\sigma +\frac{1}{2}}\right) ,  \label{rel15}
\end{equation}%
and (\ref{rel14}) follows from the threshold behaviour of $\alpha (t)$.
\hfill $\Box $ \smallskip

\begin{remark}
The threshold behaviour (\ref{rel14}) can be derived from equation (12) of
\cite{Barut/Zwanziger:1962}. The statement of Lemma 1 is therefore a general
consequence of elastic unitarity.
\end{remark}

For $4<t<5$ the logarithm is negative and diverges to minus infinity if $%
t\rightarrow 4$. Therefore the first term dominates near $t=4$. Since we
assume that $\mathrm{Im}\,\alpha (t+i0)$ is positive, there is an interval $%
\left( 4,t_{1}\right) ,\,t_{1}>4$, such that $\mathrm{Im}\,\beta (t+i0)$ is
strictly negative for $4<t<t_{1}$. The inequality $\mathrm{Im}\,\beta <0$
near threshold implies that the partial wave coefficients of $\beta $ are
negative for large angular momenta. Hence the techniques of section \ref%
{R-inel} cannot be applied to prove the inequalities (\ref{int9}) for the
fixed point solutions.

Despite of this negative statement there remains a chance to derive these
inequalities. Since the amplitude (\ref{rel3}) satisfies elastic unitarity,
there exist counterterms which modify the amplitude near the $t$-channel
threshold $t=4$ in order to achieve the correct support of the double
spectral function $\rho (s,t)$. Moreover the double spectral function has to
be positive in the Mahoux-Martin domain. The problems with $\mathrm{Abs}%
_{s}R(t,s)\simeq \beta (t)s^{\alpha (t)}$ arise exactly in a region, where
such compensations take place. With some more efforts it should be possible
to prove the existence of a solution of the fixed point equations of section %
\ref{R-fixedpoint}, which satisfies positivity for all angular momenta and
all energies. There is even a chance to find a Regge amplitude which
satisfies all constraints of section \ref{not} and which has an
asymptotically constant total cross section. But to obtain amplitudes with
increasing total cross sections one has to develop new techniques.

\appendix

\section{Estimates of partial wave amplitudes\label{PW}}

The estimates of this Appendix are based on calculations in \cite%
{Kupsch:1971, Kupsch:1982}.

\subsection{General statements\label{PW1}}

The Legendre functions of second kind $Q_{l}(z),\,l=0,1,2,...$, satisfy the
relations
\begin{equation}
\begin{array}{c}
0<Q_{l+1}(z)<Q_{l}(z)<z^{-l}Q_{0}(z)\leq Q_{0}(z)=\frac{1}{2}\log \frac{z+1}{%
z-1},\,z>1, \\
Q_{l}(z+x)<\frac{z}{z+x}Q_{l}(z)<Q_{l}(z),\;z>1,\,x>0, \\
\left( Q_{l}(z)\right) ^{2}\leq 3zQ_{0}(z)Q_{l}(2z^{2}-1),\;z>1, \\
\lim_{z\rightarrow 1+0}\;Q_{l}(z)/Q_{0}(z)=2^{-1}.%
\end{array}
\label{pw3}
\end{equation}%
If $z=1+2(s-4)^{-1}(t+v),\,s>4,\,t>0,\,v\geq 0$, then $%
2z^{2}-1>1+2(s-4)^{-1}(t+\tau _{v}(s))+16v(s-4)^{-2}t$ follows with
\begin{equation}
\tau _{v}(s):=4v+4\frac{v^{2}}{s-4}.  \label{pw4}
\end{equation}%
The function $t=\tau _{4}(s),s>4$, is the boundary of the double spectral
domain as determined by elastic unitarity. To evaluate the partial wave
amplitudes it is convenient to introduce the following functions $\Phi
_{l}^{\mu }(v;s),\,l=0,1,2,...,\,s\geq 4$, depending on the parameters $v\in %
\left[ 4,\infty \right) $ and $\mu >0$%
\begin{equation}
\begin{array}{l}
\Phi _{l}^{\mu }(v;s)=0\quad \mathrm{if}\;s=4, \\
\Phi _{l}^{\mu }(v;s)=\int_{v}^{v+1}(t-v)^{\mu }Q_{l}\left( 1+\frac{2t}{s-4}%
\right) dt=\int_{0}^{1}t^{\mu }Q_{l}\left( 1+2\frac{t+v}{s-4}\right) dt\quad
\mathrm{if}\;s>4.%
\end{array}
\label{a8}
\end{equation}%
As a consequence of (\ref{pw3}) these functions satisfy the relations%
\begin{equation}
0\leq \Phi _{l+1}^{\mu }(w;s)<\Phi _{l}^{\mu }(w;s)<\Phi _{l}^{\mu
}(v;s)\quad \mathrm{if}\;w>v\geq 4\,\mathrm{and}\,s>4.  \label{pw2}
\end{equation}%
The last of these inequalities has the following generalization: If $w>v\geq
4$, then for any pair of parameters $\mu >0,\,\sigma >0$ there is a constant
$c_{\mu \sigma }<\infty $ such that%
\begin{equation}
\Phi _{l}^{\mu }(w;s)<c_{\mu \sigma }\Phi _{l}^{\sigma }(v;s).  \label{pw2a}
\end{equation}%
For $s$ near threshold we have%
\begin{equation}
\Phi _{l}(v;s)\sim (s-4)^{l+1}\quad \mathrm{if}\;s\rightarrow 4;  \label{pw1}
\end{equation}%
and for large $s$ the functions $\Phi _{l}(v_{1};s)$ increase like $\log s$,
more precisely%
\begin{equation}
\lim_{s\rightarrow \infty }\frac{\Phi _{l}(v;s)}{\log s}=2^{-1}(1+\mu )^{-1}.
\label{pw1a}
\end{equation}

Using Schwarz' inequality we obtain $\left( \Phi _{l}(v;s)\right) ^{2}\leq
(1+\mu )^{-1}\int_{0}^{1}t^{\mu }Q_{l}^{2}\left( 1+2\frac{t+v}{s-4}\right)
dt $. Then the inequalities (\ref{pw3}) imply the following estimates for
the functions $\Phi _{l}$, uniformly in $l=0,1,2,...$,%
\begin{equation}
\left( \Phi _{l}(4;s)\right) ^{2}\leq c\cdot \frac{s-4}{s}\cdot \Phi
_{l}(20;s)\;\mathrm{if}\;4\leq s\leq 20,  \label{pw5}
\end{equation}%
or, more generally,%
\begin{equation}
\left( \Phi _{l}(v;s)\right) ^{2}\leq c\cdot \frac{s-4}{s}\log s\cdot \Phi
_{l}(\tau _{v}(w);s)\;\mathrm{if}\;4\leq s\leq w<\infty ,  \label{pw6}
\end{equation}%
and%
\begin{equation}
\left( \Phi _{l}(v;s)\right) ^{2}\leq c\cdot \frac{s-4}{s}\log s\cdot \Phi
_{l}(4v);s)\;\mathrm{if}\;s\geq 4.  \label{pw6a}
\end{equation}

The Froissart-Gribov integral (\ref{int7}) can be evaluated with the help of
the following Lemma.

\begin{lemma}
\label{FG}Let $h(t)$ be an integrable complex function on the interval $%
\left[ v,\infty \right) $ with the properties%
\[
\begin{array}{c}
\left\vert h(t)\right\vert \leq c_{1}\,(t-v)^{\mu }\;\mathrm{for}\;t\in %
\left[ v,v+1\right] \;\mathrm{with}\;\mu >0, \\
\left\vert h(t)\right\vert \leq c_{2}\,t^{\alpha }(\log t)^{-\delta }\;%
\mathrm{with}\;\alpha >-1\;\mathrm{and}\;\delta \geq 0,%
\end{array}%
\]%
where $c_{1,2}$ are positive constants, then there exists a constant $c>0$
such that
\begin{equation}
\left\vert \int_{v}^{\infty }h(t)Q_{l}\left( 1+\frac{2t}{s-4}\right)
dt\right\vert \leq \left\{
\begin{array}{c}
c\cdot s^{\alpha +1}(\log s)^{-\delta }\Phi _{l}^{\mu }(v;s)\quad \mathrm{if}%
\;\alpha >-1 \\
c\cdot \Phi _{l}^{\mu }(v;s)\quad \mathrm{if}\;\alpha =-1\;\mathrm{and}%
\;\delta >1%
\end{array}%
\right.  \label{pw7}
\end{equation}%
is valid for all $s>4$ and $l\geq \max \left\{ 0,\alpha \right\} $.
\end{lemma}

The proof of this Lemma follows from Appendix B in \cite{Kupsch:1971}. Two
important estimates for the partial wave amplitudes are again formulated as
Lemmata.

\begin{lemma}
\label{eins}Let $F(s,t)$ be an amplitude which has the properties 1. -- 3.
and 6. of section \ref{not}. The cuts in $t$ and $u$ start at $t\geq
t_{1}\geq 4$ and $u\geq t_{1}$. Then the following statement is true: For
any finite energy $s_{1}\in \left[ 4,\infty \right) $ we can find constants $%
c_{1,2}\geq 0$ such that the partial wave amplitudes of $F$ are bounded by%
\begin{eqnarray}
\left\vert f_{l}(s)\right\vert &\leq &c_{1}\cdot \left( \frac{s-4}{s}\right)
^{-\frac{1}{2}}\Phi _{l}^{\mu }(t_{1};s),  \label{pw8a} \\
\left\vert \mathrm{Im}\,f_{l}(s)\right\vert &\leq &c_{2}\cdot \Phi _{l}^{\mu
}(t_{1};s)  \label{pw8b}
\end{eqnarray}%
and%
\begin{equation}
\mathrm{Im}\,f_{l}(s)-2\left\vert f_{l}(s)\right\vert ^{2}\geq -c_{3}\cdot
s^{\gamma }\Phi _{l}^{\mu }(t_{1};s)  \label{pw8c}
\end{equation}%
for $l=0,1,2,...$ and $4\leq s\leq s_{1}$.
\end{lemma}

\begin{proof}
By assumption the amplitude $F$ is polynomially bounded, and we can use the
Froissart-Gribov integral for sufficiently large angular momenta, say $l>n$,
for all energies $s\geq 4$. The partial wave amplitudes for $l\leq n$ are
calculated with the integral (\ref{int6}). Using the Rodrigues' formula the
integral (\ref{int6}) yields $\left\vert f_{l}(s)\right\vert \leq c(l)\cdot
\left( \frac{s-4}{s}\right) ^{l+\frac{1}{2}},\,l=0,1,2,..,\,s\geq 4$, with $%
l $-dependent constants $c(l)$. This result agrees with the threshold
behaviour of (\ref{pw8a}), see (\ref{pw1}). For the finite number of angular
momenta$\,l=0,1,...,n$, and the finite energy range $4\leq s\leq s_{1}$ the $%
l$-dependent constants $c(l)$ can be absorbed into (\ref{pw8a}). For $l>n$
and $4\leq s\leq s_{1}$ the estimate (\ref{pw7}) implies a bound $\left\vert
f_{l}(s)\right\vert \leq c_{1}\cdot \sqrt{\frac{s}{s-4}}\Phi _{l}^{\mu
}(t_{1};s)$ with $\mu $ being the H\"{o}lder index of the amplitude. Hence (%
\ref{pw8a}) is valid for $4\leq s\leq s_{1}$. For the imaginary part (\ref%
{pw8b}) we get the additional threshold factor $\sqrt{\frac{s-4}{s}}$ of the
absorptive part, see (\ref{int10}). The square $\left\vert
f_{l}(s)\right\vert ^{2}$ can be calculated with the help of (\ref{pw6a})
and (\ref{pw2}). That leads to the lower bound (\ref{pw8c}). \newline
\end{proof}

\begin{lemma}
\label{zwei}Let $F(s,t)$ be an amplitude which has the properties 1. -- 3.
and 6. of section \ref{not}.. The cuts in $t$ and $u$ start at $t\geq
t_{1}\geq 4$ and $u\geq t_{1}$. If $F$ is bounded by%
\begin{equation}
\left\vert F(s+i0,t)\right\vert \leq const\cdot \sum_{j=1,2}s^{\beta
_{j}}\left( 1+\left\vert t\right\vert \right) ^{\alpha _{j}}  \label{pw9}
\end{equation}%
for $s\geq s_{1}\geq 4$ with $\alpha _{j}>-1$ and $\beta _{j}\in \mathbb{R}$%
, then there exists a constant $c_{1}>0$ such that the partial wave
amplitudes $f_{l}$ of $F$ have the upper bound%
\begin{equation}
\left\vert f_{l}(s)\right\vert \leq c_{1}\cdot \sqrt{\frac{s}{s-4}}s^{\gamma
}\Phi _{l}^{\mu }(t_{1};s),\;l=0,1,2,...,\;s\geq 4,  \label{pw10}
\end{equation}%
with $\gamma =\max_{j}\left\{ \alpha _{j}+\beta _{j}\right\} $. If $\gamma
<0 $ and $0<\mu \leq \frac{1}{2}$ then there exists a constant $c_{2}>0$
such that the inequalities%
\begin{equation}
\mathrm{Im}\,f_{l}(s)-2\left\vert f_{l}(s)\right\vert ^{2}\geq -c_{2}\cdot
s^{\gamma }\Phi _{l}^{\mu }(t_{1};s),\;l=0,1,2,...,  \label{pw11}
\end{equation}%
are true for $s\geq 4$.
\end{lemma}

\begin{proof}
By assumption the amplitude $F$ is polynomially bounded, and we can use the
Froissart-Gribov integral for sufficiently large angular momenta $l>n$ for
all energies. The partial wave amplitudes for $l\leq n$ are calculated with
the integral (\ref{int6}). \newline
For the finite energy range $4\leq s\leq \max \left\{ 5,s_{1}\right\} $ the
estimate (\ref{pw10}) follows from Lemma \ref{eins}. If $s\geq \max \left\{
5,\,s_{1}\right\} $ the integral (\ref{int6}) yields the uniform bound $%
\left\vert f_{l}(s)\right\vert \leq c_{1}\cdot s^{\gamma }$ for all $%
l=0,1,2,...$. This bound together with the estimate derived for the partial
wave amplitudes $f_{l}(s),\,l>n$, with Lemma \ref{FG} imply the upper bound (%
\ref{pw10}) for $s\geq \max \left\{ 5,\,s_{1}\right\} $.\newline
The imaginary parts of $f_{l}$ has the additional threshold factor $\sqrt{%
\frac{s-4}{s}}$ of the estimate (\ref{pw8b}). The square $\left\vert
f_{l}(s)\right\vert ^{2}$ can be calculated with the help of (\ref{pw6a})
and (\ref{pw2}). That leads to the lower bound (\ref{pw11}).
\end{proof}

\subsection{Amplitudes with positive spectral functions\label{PW2}}

Estimates of the partial wave amplitudes of a Mandelstam representation with
positive spectral functions are needed for the construction of amplitudes,
which satisfy the inelastic unitarity inequalities (\ref{int9}) for all
energies.

It is straightforward to derive a precise estimate of the partial wave
amplitudes of $A(s,t)=Sym\,F(s,t)$, where $F$ satisfies the Mandelstam
representations (\ref{a1}) or (\ref{a2}), if the spectral functions have the
following properties:

\begin{itemize}
\item The double spectral function $\rho (x,y)$ is H\"{o}lder continuous
with index $\mu \in \left( 0,\frac{1}{2}\right] $. It has the structure%
\begin{equation}
\begin{array}{c}
\rho (s,t)=\psi _{1}(s,t)+\psi _{2}(t,s)\quad \mathrm{with} \\
0\leq \psi _{1,2}(s,t)\leq c\cdot t^{\alpha }s^{-1}(\log t)^{-\delta }(\log
s)^{-\delta }%
\end{array}
\label{pw21}
\end{equation}%
with $-1<\alpha \leq 1,\,\delta >1$. The support of $\psi _{1,2}(s,t)$ lies
within $\left[ v,\infty \right) \times \left[ w,\infty \right) $ with $4\leq
v\leq w\leq 20$.
\end{itemize}

In the sequel we simply write $\psi $ instead of $\psi _{1}$ and $\psi _{2}$%
. In the case of the crossing symmetry (\ref{int3}) of neutral pions we
anyhow have $\psi _{1}(s,t)=\psi _{2}(s,t)$. But the following estimates are
also valid for isospin-1 pions. If $-1<\alpha \leq 0$ we can use the
unsubtracted Mandelstam representation (\ref{a1}). If $0<\alpha \leq 1$ we
have to take the subtracted Mandelstam representation (\ref{a2}) which has
an additional single spectral function with the properties:

\begin{itemize}
\item The single spectral function $\varphi (s)$ is H\"{o}lder continuous
with index $\mu \in \left( 0,\frac{1}{2}\right] $, positive and bounded by $%
s^{\alpha -1}\left( \log s\right) ^{-\delta },\,0<\alpha \leq 1,\,\delta >1$%
, for large $s$. It has a threshold behaviour $\varphi (s)\geq
c(s-v)^{\sigma }$ with $c>0$ for $v\leq s\leq v+1,\,v\geq 4$. The exponent
is $\sigma =\frac{1}{2}$ if $v=4$, and $\sigma =\mu \leq \frac{1}{2}$ if $%
v>4 $.
\end{itemize}

The Hilbert transform in (\ref{a1}) or (\ref{a2}) introduce additional $\log
s$ and $\log t$ factors. The partial wave (\ref{int6}) $a_{0}(s)$ has the
bound $\left\vert a_{0}(s)\right\vert \leq c\sqrt{\frac{s-4}{s}}s^{\alpha
-1}(\log s)^{-\delta +2}$. All higher partial wave amplitudes can be
estimated with the Froissart-Gribov integral (\ref{int7}) using Lemma \ref%
{FG}. The final result -- including all crossed terms -- is
\begin{equation}
\left\vert a_{l}(s)\right\vert \leq c\sqrt{\frac{s}{s-4}}s^{\alpha -1}(\log
s)^{-\delta +1}\Phi _{l}^{\mu }(v;s),\quad l=0,1,2,...,s>4.  \label{pw22}
\end{equation}%
In the case of one subtraction we have $0<\alpha \leq 1$, in the case
without subtraction the value of $\alpha $ is $-1<\alpha \leq 0$. The
relations (\ref{pw6}) and (\ref{pw6a}) imply the upper bounds%
\begin{equation}
\left\vert a_{l}(s)\right\vert ^{2}\leq \left\{
\begin{array}{l}
c\cdot \Phi _{l}^{\mu }(20;s)\quad \mathrm{if}\;4\leq s\leq 20, \\
c\cdot s^{2\alpha -2}(\log s)^{-2\delta +2}\Phi _{l}^{\mu }(4v;s)\quad
\mathrm{if}\;s\geq 4.%
\end{array}%
\right.  \label{pw23}
\end{equation}

For a lower bound on the imaginary part we need a more detailed knowledge
about the behaviour of the double spectral function near the boundary of its
support. If the support of $\psi (s,t)$ starts at $s=v$ and $t=w$, and $\psi
$ is bounded from below by%
\begin{equation}
\psi (s,t)\geq c\cdot \left( \frac{s-v}{s}\right) _{+}^{\sigma }\left( \frac{%
t-w}{t}\right) _{+}^{\mu }t^{\alpha }s^{-1}(\log t)^{-\delta }(\log
s)^{-\delta }\quad \mathrm{if\,either}\;v\leq s\leq v+1\;\mathrm{or}\;w\leq
t\leq w+1,  \label{pw24}
\end{equation}%
then the Froissart-Gribov integral implies that the imaginary parts of the
partial wave amplitudes with $l>0$ have the lower bound
\begin{equation}
\left( \frac{s-4}{s}\right) ^{\frac{1}{2}}\mathrm{Im}\,a_{l}(s)\geq
c_{1}\left( \frac{s-v}{s}\right) _{+}^{\sigma }s^{-2}(\log s)^{-\delta }\Phi
_{l}^{\mu }(w;s)+c_{2}\left( \frac{s-w}{s}\right) _{+}^{\mu }s^{\alpha
-1}(\log s)^{-\delta }\Phi _{l}^{\sigma }(v;s)  \label{pw25}
\end{equation}%
for $s\geq 4$. Here $c$ and $c_{1,2}$ are positive constants. This lower
bound is also correct for $l=0.$

Since $\Phi _{l}^{\mu }(v;s)\sim \log s$ for large $s$, the partial wave
amplitudes have the asymptotic behaviour $\left\vert \mathrm{Re}%
\,a_{l}(s)\right\vert \lesssim s^{\alpha -1}(\log s)^{-\delta +2}$ and $%
\mathrm{Im}\,a_{l}(s)\sim s^{\alpha -1}(\log s)^{-\delta +1}$. This
statement implies

\begin{corollary}
\label{delta} For increasing amplitudes with $\alpha =1$ the inequalities $%
\left\vert a_{l}(s)\right\vert ^{2}\leq \mathrm{Im}\,a_{l}(s)$ can be
derived only if $\delta \leq 3$. To obtain values $\delta <3$ one needs an
improved estimate for $\mathrm{Re}\,a_{l}(s)$.
\end{corollary}

Now we take thresholds $v=4$ and $4\leq w\leq 20$ and choose parameters $%
\alpha <1$ and $\delta >1$ or $\alpha =1$ and $\delta \geq 3$. The estimates
(\ref{pw23}) and (\ref{pw25}) imply that the partial wave amplitudes satisfy
the inequalities%
\begin{equation}
\lambda \cdot \mathrm{Im}\,\,a_{l}(s)-\left\vert a_{l}(s)\right\vert
^{2}\geq 0,\;s\geq 4,\,l=0,1,2,...  \label{pw27}
\end{equation}%
if the multiplier $\lambda >0$ is large enough. Then the partial wave
amplitudes of $\lambda ^{-1}A(s,t)$ fulfil the unitarity relations (\ref%
{int9}) for $s\geq 4$. Choosing a larger value of $\lambda $ we get%
\begin{eqnarray}
\lambda \cdot \mathrm{Im}\,a_{l}(s)-\left\vert a_{l}(s)\right\vert ^{2}
&\geq &\left\{
\begin{array}{c}
c_{1}\cdot \Phi _{l}^{\mu }(w;s)\quad \mathrm{if}\;4\leq s\leq w+1, \\
c_{2}\cdot s^{\alpha -1}(\log s)^{-\delta }\Phi _{l}^{\mu }(v;s)\quad
\mathrm{if}\;s\geq v+1%
\end{array}%
\right.  \label{pw28} \\
&\geq &c\cdot s^{\alpha -1}(\log s)^{-\delta }\Phi _{l}^{\mu }(w;s)\quad
\mathrm{if}\;s\geq 4.  \label{pw29}
\end{eqnarray}%
for $l=0,1,2,..$ with constants $c_{1,2}>0$ and $c>0$ such that the partial
wave amplitudes of $(2\lambda )^{-1}A(s,t)$ satisfy the relations%
\begin{equation}
\mathrm{Im}\,a_{l}(s)-2\left\vert a_{l}(s)\right\vert ^{2}\geq c_{0}\cdot
\left( \frac{s-4}{s}\right) ^{\mu -\frac{1}{2}}s^{\alpha -1}(\log
s)^{-\delta }\Phi _{l}^{\mu }(w;s)  \label{pw31}
\end{equation}%
for $s\geq 4$ and $l=0,1,2,...$, where $c_{0}>0$ is a strictly positive
constant.

\begin{remark}
\label{unsubt}In the case of unsubtracted amplitudes we can write down
amplitudes with partial wave amplitudes, which behave like (\ref{pw22}) and (%
\ref{pw25}) for all $l=0,1,2,...$, where the parameter $\alpha $ lies in the
interval $-2<\alpha <0$. See the end of section 2.2 of \cite{Kupsch:1971}.
\end{remark}

The double spectral function $\omega (s,t)$ of the inhomogeneous part $%
B(s,t) $ in section \ref{fixedpoint} has thresholds at $v=w=16$. If $\omega
(s,t)$ has a lower bound (\ref{pw24}) (with $\alpha \leq 0$ in the case
without subtraction), the imaginary parts of the partial wave amplitudes $%
b_{l}(s)$ of $B$ have the lower bounds%
\begin{equation}
\mathrm{Im}\,\,b_{l}(s)\geq c_{1}\left( \frac{s-16}{s}\right) _{+}^{\mu
}s^{\alpha -1}(\log s)^{-\delta }\Phi _{l}^{\mu }(16;s)\;\mathrm{if}\;s\geq
4,  \label{pw32}
\end{equation}%
and they satisfy the inequalities
\begin{equation}
\mathrm{Im}\,b_{l}(s)\geq c_{2}\left\vert b_{l}(s)\right\vert ^{2}\quad
\mathrm{if}\;s\geq 17.  \label{pw33}
\end{equation}%
These inequalities are sufficient to prove that -- after an appropriate
scaling of the inhomogeneous part $\omega $ -- the fixed point problem of
section \ref{fixedpoint} yields amplitudes, which also satisfy the unitarity
inequalities (\ref{int9}) for all energies $s\geq 16$.

\section{Calculations for the Regge pole amplitudes\label{R}}

\subsection{Functions with positive Legendre/Taylor coefficients\label%
{PWpositive}}

In section \ref{lin} the set $\mathcal{A}(s_{1}),\,s_{1}\geq 4$, has been
defined as the set of\textit{\ }functions $\phi (s,t)$, which have positive
coefficients in their Legendre (partial wave) expansion for all energies $%
s\geq s_{1}$. In this Appendix this set will be called $\mathcal{A}%
_{1}(s_{1})$. For further calculations it is useful to define a subclass $%
\mathcal{A}_{2}(s_{1})\subset \mathcal{A}_{1}(s_{1})$, which is
characterized by positive coefficients in the Taylor expansion with respect
to the variable $z=1+2(s-4)^{-1}t$.

A function $\phi (s,t)$ is an element of the set $\mathcal{A}%
_{2}(s_{1}),~s_{1}\geq 4$, if it has the following properties:

\begin{enumerate}
\item The function $\phi (s,t)$ is H\"{o}lder continuous in $s\geq s_{1}\geq
4$ and holomorphic in the variable $t$ in the cut plane $\mathbb{C}_{cut}$.

\item The function $\phi (s,t)$ is real if $s\geq s_{1}$ and $-s<t<4$.

\item The power series expansion $\phi (s,t)=\sum c_{n}(s)x^{n}$ with the
variable $x=\frac{s-4}{2}z=t+\frac{s-4}{2}$ has positive coefficients $%
c_{n}(s)\geq 0$ for $n=0,1,2,...$ and $s\geq s_{1}$.
\end{enumerate}

In \cite{Kupsch:1971} the class $\mathcal{A}_{2}$ has been called $\mathcal{A%
}^{\prime }$. Since $P_{n}(z)=\sum_{k=0}^{n}a_{k}z^{k}$ with positive
coefficients $a_{k}$ we have $\mathcal{A}_{2}(s_{1})\subset \mathcal{A}%
_{1}(s_{1})$; and the inclusions $\mathcal{A}_{1,2}(s_{2})\subset \mathcal{A}%
_{1,2}(s_{1})$ if $s_{1}\leq s_{2}$ are obvious. The algebraic structures of
the spaces $\mathcal{A}_{1}(s_{1})$ and $\mathcal{A}_{2}(s_{1})$ are similar:

From $\phi _{1,2}(s,t)\in \mathcal{A}_{k},\,k=1,2$, we get
\begin{eqnarray}
\alpha \phi _{1}(s,t)+\beta \phi _{2}(s,t) &\in &\mathcal{A}_{k}\quad
\mathrm{if}\;\alpha ,\beta \geq 0  \label{p1} \\
\phi _{1}(s,t)\cdot \phi _{2}(s,t) &\in &\mathcal{A}_{k}\,,  \label{p2}
\end{eqnarray}%
and $\phi (s,t)\in \mathcal{A}_{k}$ implies $\exp \left( \lambda \phi
(s,t)\right) \in \mathcal{A}_{k}$ for all parameters $\lambda \geq 0$.

Let $f(t)$ be a real function, which is $\mathcal{L}^{p}$-integrable over
the interval $4\leq t<\infty $ for some $p\in (1,\infty )$. Then the
analytic function $F(t)=\int_{4}^{\infty }f(t^{\prime })\left( t^{\prime
}-t\right) ^{-1}dt^{\prime }$ has a well defined power series expansions $%
F(t)=\sum_{n=0}^{\infty }c_{n}^{F}(s)x^{n}$ with coefficients%
\begin{equation}
c_{n}^{F}(s)=\frac{1}{\pi }\int_{4}^{\infty }f(t)\left( \frac{s-4}{2}%
+t\right) ^{-n-1}dt  \label{p5}
\end{equation}%
If $f(t)\geq 0$ then these coefficients are positive, and the function $F(t)$
is an element of $\mathcal{A}_{2}(4)$. As a consequence the trajectory
function (\ref{r5}) $\alpha (t)$ is an element of $\mathcal{A}_{2}(4)$ if $%
\alpha (\infty )\geq 0$. In the general case we have $\alpha (t)=\alpha
_{0}+\alpha _{1}(t)$ with $\alpha _{0}\leq 0$ and $\alpha _{1}(t)\in
\mathcal{A}_{2}(4)$.

By power series expansion we obtain%
\begin{equation}
-\log (t_{1}-t)+c\log (u_{1}-u)\in \mathcal{A}_{2}(4)  \label{p7}
\end{equation}%
where $u=4-s-t$, and the parameters are restricted to $4\leq t_{1}\leq u_{1}$
and $-1\leq c\leq 1$. As a consequence of (\ref{p7}) and the rules (\ref{p1}%
) and (\ref{p2}) we obtain the following examples, which are needed for
section \ref{risingR}. Thereby the function $\gamma _{1}(t)$ is an element
of $\mathcal{A}_{2}(4)$, and $\gamma (t)$ is the sum $\gamma (t)=-\delta
+\gamma _{1}(t)$ with a constant $\delta \geq 0$.
\begin{eqnarray*}
&&(s-s_{1})^{\gamma (t)}=(s-s_{1})^{-\delta }\exp \left( \gamma _{1}(t)\log
(s-s_{1})\right) \in \mathcal{A}_{2}(s_{1}+2),\,s_{1}\geq 4, \\
&&(t_{1}-t)^{-2\delta -\gamma (t)}=\exp \left( -(\delta +\gamma
_{1}(t))(\log (t_{1}-t))\right) \in \mathcal{A}_{2}(4), \\
&&(t_{1}-t)^{-\delta }(u_{1}-u)^{-\delta }\in \mathcal{A}_{2}(4), \\
&&(t_{1}-t)^{-\gamma _{1}(t)}(u_{1}-u)^{\gamma _{1}(t)}\in \mathcal{A}%
_{2}(4), \\
&&(t_{1}-t)^{-\gamma (t)-2\delta }(u_{1}-u)^{\gamma (t)}\in \mathcal{A}%
_{2}(4).
\end{eqnarray*}%
These results can be extended to integrals with positive weight functions.
Let $\sigma (t)\geq 0$ be a positive integrable function with support inside
the interval $17\leq t\leq 18$. Then we have $\int \sigma (s^{\prime
})(s-s^{\prime })^{\gamma (t)}ds^{\prime }\in \mathcal{A}_{2}(20)$ and%
\begin{equation}
N_{1}(s,t):=\int \int dt^{\prime }du^{\prime }\sigma (t^{\prime })\sigma
(u^{\prime })(t^{\prime }-t-1)^{-\gamma (t)-2\delta }(u^{\prime
}-4+s+t)^{\gamma (t)}\in \mathcal{A}_{2}(4).  \label{p9}
\end{equation}%
This result has an important consequence for section \ref{risingR}. If we
factorize the residue function into $\beta (t)=\beta _{1}(t)\cdot \beta
_{0}(t)$ with $\beta _{0}(t)=\int dt^{\prime }\sigma (t^{\prime })(t^{\prime
}-t-1)^{-\gamma (t)-2\delta }\in \mathcal{A}_{2}(4)$, then the imaginary
part (\ref{r8}) of the Regge ansatz is the product%
\begin{equation}
\mathrm{Im}\,\hat{R}(s+i0,u)=\beta _{1}(t)\cdot \int ds^{\prime }\sigma
(s^{\prime })(s-s^{\prime })_{+}^{\gamma (t)}\cdot N_{1}(s,t).  \label{p10}
\end{equation}%
Hence choosing $\beta _{1}(t)$ as element of $\mathcal{A}_{2}(4)$ or of the
larger class $\mathcal{A}_{1}(4)$ the product (\ref{p10}) is an element of $%
\mathcal{A}_{1}(20)$ as stated in section \ref{risingR}.

\begin{remark}
If the trajectory $\alpha (t)$ enters the half plane $\left\{ l\mid \mathrm{%
Re}\,l<0\right\} $ for $t<0$, the factor $\left( \sin \pi \gamma (t)\right)
^{-1}$ in (\ref{r4}) has poles in the physical region, and one needs zeros
of $\beta _{1}(t)$ to compensate these poles. Such functions $\beta _{1}(t)$
do not exist within $\mathcal{A}_{2}(4)$ but in the larger class $\mathcal{A}%
_{1}(4)$, see Appendix D of \cite{Kupsch:1971}. For the proof of the
unitarity inequalities for the amplitudes, which saturate the Froissart
bound, one has also to work with functions of the class $\mathcal{A}_{1}$,
see \cite{Kupsch/Pool:1979,Kupsch:1982}.
\end{remark}

For the comparison of real analytic functions, which are defined by
dispersion integrals, the following Lemma is useful.

\begin{lemma}
\label{compare}Let $f(t)$ and $g(t)$ be two real functions on the interval $%
4\leq t<\infty $, which are $\mathcal{L}^{p}$-integrable with $1<p<\infty $
. Assume that these functions satisfy the following restrictions:\newline
a) The function $f(t)$ is positive, $f(t)\geq 0$, and the threshold $t=4$
belongs to the support of $f(t)$. \newline
b) There exists a constant $c_{1}>0$ such that $\left\vert g(t)\right\vert
\leq c_{1}f(t)$ in an interval $4\leq t\leq t_{1}$ and for large $%
t,\,t>t_{2}\geq t_{1}$. \newline
Then there exists a constant $c_{2}\geq 0$ such that the power series
coefficients (\ref{p5}) of $F(t)=\int_{4}^{\infty }f(t^{\prime })\left(
t^{\prime }-t\right) ^{-1}dt^{\prime }$ and $G(t)=\int_{4}^{\infty
}g(t^{\prime })\left( t^{\prime }-t\right) ^{-1}dt^{\prime }$ satisfy $%
\left\vert c_{n}^{G}(s)\right\vert \leq c_{2}\cdot c_{n}^{F}(s)$ for $%
n=0,1,2,...$ and $s>4$.
\end{lemma}

The proof follows from the representation (\ref{p5}); see Appendix D of \cite%
{Kupsch:1971}.

\subsection{Unitarity of the Regge ansatz\label{unitRegge}}

For the constructions presented here the trajectories $\alpha (t)$ and the
residue functions $\beta (t)$ have the following properties, which are more
restrictive than those given in section \ref{R-inel}; for the general case
see \cite{Kupsch:1971}.

\begin{enumerate}
\item[a)] The trajectory $\alpha (t)$ is real analytic and satisfies the
dispersion relation (\ref{r5}). The imaginary part is H\"{o}lder continuous
and strictly positive $\mathrm{Im}\,\alpha (x+i0)>0$ for $x>4$. At threshold
the imaginary part is bounded by $\mathrm{Im}\,\alpha (t+i0)\leq c\sqrt{t-4}$
if $4\leq t\leq 5$ with $c>0$.

\item[b)] The values of $\alpha (t)$ are restricted by $\alpha (0)\leq 1$
and $\,0<\alpha (\infty )<\alpha (t)<2$ if $t\leq 4$.

\item[c)] The residue function $\beta (t)$ factorizes into%
\begin{equation}
\beta (t)=\beta _{1}(t)\cdot \beta _{0}(t)\;\mathrm{with}\;\beta
_{0}(t)=\int dt^{\prime }\sigma (t^{\prime })(t^{\prime }-t-1)^{-\gamma (t)}.
\label{r5a}
\end{equation}%
The convolution is performed with the same function $\sigma (t)$ as used in (%
\ref{r4}). The function $\beta _{1}(t)$ is real analytic with a cut at $%
t\geq 4$ and a H\"{o}lder continuous absorptive part. It has positive
partial wave amplitudes and $\beta _{1}(t)$ is bounded by $\left\vert \beta
_{1}(t)\right\vert \leq const\,\left( 1+\left\vert t\right\vert \right)
^{-\delta },\,t\in \mathbb{C}_{cut}$, with $\delta >\frac{1}{2}$.
\end{enumerate}

The restriction $\alpha (\infty )>0$ in b) allows only trajectories which
stay in the right half plane $\mathrm{Re}\,l>0$ below threshold $t=4$. This
assumption simplifies the subsequent arguments. The factor $\beta _{0}(t)$
in (\ref{r5a}) is an element of the class $\mathcal{A}_{2}(4)$, see Appendix %
\ref{PWpositive}. With $\beta _{1}(t)\in \mathcal{A}_{1}(4)$ the
residue function (\ref{r5a}) has positive partial wave coefficients,
as assumed in section \ref{R-inel}.

As a consequence of these assumptions the Regge ansatz (\ref{r4}) $\hat{R}%
(s,u)$ has the upper bound%
\begin{equation}
\left\vert \hat{R}(s,u)\right\vert \leq c\cdot \left( 1+\left\vert
t\right\vert \right) ^{-\theta }\left( 1+\left\vert s\right\vert \right)
^{\varpi (t)}\left( 1+\left\vert s-t\right\vert \right) ^{\varpi
(t)},\;(s,t)\in \mathbb{C}_{cut}^{2},  \label{r20}
\end{equation}%
with $\varpi (t)=\mathrm{Re}\,\gamma (t+i0)$ and an exponent $\theta
>2^{-1}\left( 1+\alpha (\infty )\right) $. For fixed $t$ the crossed
contributions $\hat{R}(s,t)+\hat{R}(t,u)$ decrease stronger than $s^{-\frac{1%
}{2}}$ if $s\rightarrow \infty $. The large $s$ asymptotics is therefore
dominated by $\hat{R}(s,u)$. Since $\alpha (\infty )>0$ the background
contribution $G(s,t)$ is chosen to satisfy a once subtracted Mandelstam
representation, see section \ref{crossing}.

The estimate (\ref{r20}) implies the uniform bound (\ref{r6}) for the
partial wave amplitudes. Since we are interested in trajectories with $%
\max_{t\geq 4}\mathrm{Re}\,\alpha (t+i0)>1$ the Froissart-Gribov integral
does not give good estimates for the $l$-dependence of the partial wave
amplitudes at high energies. But within a finite energy range, say $4\leq
s\leq 20$, we obtain from Lemma \ref{FG} of Appendix \ref{PW1}%
\begin{equation}
\left\vert f_{l}(s)\right\vert \leq c_{2}\sqrt{\frac{s}{s-4}}\Phi
_{l}(4;s),\;4\leq s\leq 20.  \label{r7}
\end{equation}

In the next step the linear unitarity relations (\ref{lin5}) of section \ref%
{lin} are derived for the Regge ansatz (\ref{r4}). If $s\geq 18$ (and $%
4-s\leq t\leq 0$) the imaginary part $N(s,t):=\mathrm{Im}\,R(s+i0,u)$ and
the real part $M(s,t):=\mathrm{Re}\,R(s+i0,u)$ are related by, see (\ref{r8}%
) and (\ref{r9}),
\begin{equation}
M(s,t)=-\cot \pi \gamma (t)\cdot N(s,t),  \label{r10}
\end{equation}%
and we have%
\begin{equation}
N(s,t)=0\quad \mathrm{if}\quad 4\leq s\leq 17.  \label{r10a}
\end{equation}%
As a consequence of property b) of the trajectory the function $\cot \pi
\gamma (t)$ is holomorphic for $t\in \mathbb{C}\backslash \left[ 4,\infty
\right) $, and the imaginary part has the upper bound
\begin{equation}
\left\vert \mathrm{Im}\cot \pi \gamma (t+i0)\right\vert \leq c\cdot \mathrm{%
Im}\,\gamma (t+i0)\;\mathrm{if}\;t>4  \label{r11}
\end{equation}%
with some constant $c\geq 0$. The function $\beta (t)$ is now factorized
into (\ref{r5a}). Then (\ref{r8}) and (\ref{r9}) can be written as
\begin{equation}
\begin{array}{l}
N(s,t)=\beta _{1}(t)\cdot \int ds^{\prime }\sigma (s^{\prime })(s-s^{\prime
})_{+}^{\gamma (t)}\cdot N_{1}(s,t)\quad \mathrm{and} \\
M(s,t)=-\beta _{1}(t)\cdot \cot \pi \gamma (t)\cdot \int ds^{\prime }\sigma
(s^{\prime })(s-s^{\prime })_{+}^{\gamma (t)}\cdot N_{1}(s,t)%
\end{array}
\label{r13}
\end{equation}%
where $N_{1}(s,t)\in \mathcal{A}_{2}(4)$ is given by (\ref{p9}) with $\delta
=0$. The imaginary part $N(s,t)$ is an element of $\mathcal{A}_{2}(20)$.
Moreover, with the help of Lemma \ref{compare} we can find a function $\beta
_{1}(t)\in \mathcal{A}_{2}(4)$ with positive imaginary part such that%
\begin{equation}
-c\cdot \beta _{1}(t)\prec -\beta _{1}(t)\cot \pi \gamma (t)\prec c\cdot
\beta _{1}(t)  \label{r15}
\end{equation}%
holds with some constant $c>0$. Using (\ref{lin4}) these relations imply
that the partial wave amplitudes of $M(s,t)$ can be estimated by those of $%
N(s,t)\in \mathcal{A}(20)$%
\begin{equation}
-c\cdot N(s,t)\prec M(s,t)\prec c\cdot N(s,t)\;\mathrm{if}\;s\geq 20.
\label{r17}
\end{equation}

Following the arguments of section \ref{lin} the partial wave amplitudes of
the Regge ansatz $R(s,u)$ satisfy the quadratic unitarity inequalities (\ref%
{lin2}) for $s\geq 20$. For energies $4\leq s\leq 17$ we have $\mathrm{Im}%
\,f_{l}(s)=0$ and -- using (\ref{r7}) and (\ref{pw5}) -- we obtain the upper
bound $\left\vert f_{l}(s)\right\vert ^{2}\leq c\sqrt{\frac{s-4}{s}}\Phi
_{l}(20;s)$. These results imply the lower bounds
\begin{equation}
\mathrm{Im}\,f_{l}(s)-c\left\vert f_{l}(s)\right\vert ^{2}\geq \left\{
\begin{array}{l}
-c_{1}\cdot \Phi _{l}(20;s)\quad \mathrm{if}\;4\leq s\leq 17, \\
-c_{2}\cdot \Phi _{l}(4;s)\quad \mathrm{if}\;17\leq s\leq 20, \\
0\quad \mathrm{if}\;s\geq 20,%
\end{array}%
\right.  \label{r18}
\end{equation}%
with some constants $c_{1,2}>0$. The partial wave amplitudes of the crossed
term $R(s,t)$ are $(-1)^{l}f_{l}(s)\,,l=0,1,2,...$. Hence the partial wave
amplitudes of the sum $R(s,u)+R(s,t)$ are $2f_{l}(s)$ if $l$ is even, and $0$
if $l$ is odd. These partial wave amplitudes satisfy again an estimate of
the type (\ref{r18}).

\subsection{Crossing symmetry\label{crossing}}

Crossing symmetry, correct threshold behaviour and the inelastic unitarity
inequalities (\ref{int9}) for all energies $s\geq 4$ can be incorporated
with a method which has been developed in \cite{Kupsch:1971,Kupsch:1982}.
The main results can be summarized in the following Propositions.

\begin{proposition}
\label{drei}Let $F(s,t)=F(s,u)$ be an amplitude, which is symmetric in $t$
and $u$ and which has the properties 1. -- 3. and 6. of section \ref{not}.
Assume there exists an energy $s_{1}\geq 4$ such that $F(s,t)$ has the upper
bound%
\[
\left\vert F(s+i0,t)\right\vert \leq const\cdot \sum_{j=1,2}s^{\beta
_{j}}\left( 1+\left\vert t\right\vert \right) ^{\alpha _{j}}
\]%
for $s\geq s_{1}\geq 4$ with $\alpha _{j}>-1$ and $\alpha _{j}+\beta _{j}<0$
(or $\alpha _{j}+\beta _{j}<-1$),$\;j=1,2$. Then one can find a constant $%
\lambda >0$ and a crossing symmetric amplitude $G(s,t)$, which satisfies a
Mandelstam representation with at most one (without) subtraction and with
positive spectral functions, such that the following statement is true:%
\newline
The sum $A(s,t)=\lambda F(s,t)+G(s,t)$ fulfils the unitarity inequalities $%
\mathrm{Im}\,a_{l}(s)\geq \left\vert a_{l}(s)\right\vert ^{2},\;l=0,1,2,...$%
for all energies $s\geq 4$.
\end{proposition}

\begin{proof}
Following Lemma \ref{zwei} in Appendix \ref{PW1} the partial wave amplitudes
of $F(s,t)$ can be estimated by (\ref{pw11}) with $\gamma =\max_{j}\left\{
\alpha _{j}+\beta _{j}\right\} <0$ and a constant $c_{2}>0$. Let $G(s,t)$ be
a crossing symmetric amplitude, which has a Mandelstam representation with
at most one subtraction and with positive spectral functions (as considered
in Appendix \ref{PW2}). We assume that the inequalities (\ref{pw31}) for the
partial wave amplitudes $\mathrm{Im}\,g_{l}(s)-2\left\vert
g_{l}(s)\right\vert ^{2}\geq c_{0}\cdot \left( \frac{s-4}{s}\right) ^{\mu -%
\frac{1}{2}}s^{\alpha -1}(\log s)^{-\delta }\Phi _{l}^{\mu }(t_{1};s)$ with $%
\alpha $ in the interval $\gamma +1<\alpha <1$ are valid. Then the amplitude
$A=\lambda F+G$ with $0<\lambda \leq \min \left\{ 1,c_{2}^{-1}c_{0}\right\} $
has partial wave amplitudes which fulfil the constraints
\[
\begin{array}{c}
\mathrm{Im}\,(\lambda f_{l}+g_{l})-\left\vert \lambda f_{l}+g_{l}\right\vert
^{2}\geq \lambda \mathrm{Im}\,f_{l}-2\lambda ^{2}\left\vert f_{l}\right\vert
^{2}+\mathrm{Im}\,g_{l}-2\left\vert g_{l}\right\vert ^{2} \\
\geq \lambda \mathrm{Im}\,f_{l}-2\lambda \left\vert f_{l}\right\vert ^{2}+%
\mathrm{Im}\,g_{l}-2\left\vert g_{l}\right\vert ^{2}\geq 0,\;l=0,2,4,...%
\end{array}%
\]%
\ for $s\geq 4$. If $\gamma <-1$ then $\gamma +1<\alpha <0$ is possible, and
we can choose an amplitude $G$ which satisfies an unsubtracted Mandelstam
representation.
\end{proof}

\begin{proposition}
\label{vier}Let $F(s,t)=F(s,u)$ be an amplitude, which is symmetric in $t$
and $u$ and which has the properties 1. -- 3. and 6. of section \ref{not}.
Assume there exists an energy $s_{1}\geq 4$ and a constant $\lambda _{0}>0$
such that the partial wave amplitudes of $F$ satisfy the inelastic unitarity
constraints $\mathrm{Im}f_{l}(s)\geq \lambda _{0}\left\vert
f_{l}(s)\right\vert ^{2}$ for $s\geq s_{1}$. Then we can find a constant $%
\lambda >0$ and a crossing symmetric amplitude $G(s,t)$, which satisfies an
unsubtracted Mandelstam representation with positive double spectral
function, such that the sum $A(s,t)=\lambda F(s,t)+G(s,t)$ fulfils the
unitarity inequalities $\mathrm{Im}a_{l}(s)\geq \left\vert
a_{l}(s)\right\vert ^{2},\;l=0,1,2,...$for all energies $s\geq 4$.
\end{proposition}

\begin{proof}
The amplitude $F(s,t)$ is polynomially bounded and has a threshold behaviour
(\ref{int10}). The partial wave amplitudes satisfy the estimates of Lemma %
\ref{eins} for $s\leq s_{1}$, hence%
\begin{eqnarray*}
\mathrm{Im}\,f_{l}-\lambda _{0}\left\vert f_{l}\right\vert ^{2} &\geq
&\left\{
\begin{array}{l}
-c\cdot \Phi _{l}^{\mu }(t_{1};s)\quad \mathrm{if}\;4\leq s\leq s_{1} \\
0\quad \mathrm{if}\;s\geq s_{1}%
\end{array}%
\right. \\
&\geq &-c_{\alpha }\cdot s^{\alpha -1}\Phi _{l}^{\mu }(t_{1};s)\quad \mathrm{%
if}\;s\geq 4,
\end{eqnarray*}%
where any constant $\alpha <0$ is possible. Using the arguments of the proof
for Proposition \ref{drei} we can find a constant $\lambda >0$ and an
amplitude $G(s,t)$, which is given by an unsubtracted Mandelstam
representation, such that the partial wave amplitudes $a_{l}=\lambda
f_{l}+g_{l}$ satisfy the unitarity inequalities for all $s\geq 4$.
\end{proof}

We now apply these Propositions to the Regge amplitudes of section \ref%
{Regge}. The $t-u$ symmetrized Regge contribution $F(s,t):=\hat{R}(s,u)+\hat{%
R}(s,t)$ satisfies the assumptions of Proposition \ref{vier}. Hence we can
find a constant $\lambda _{1}>0$ and a crossing symmetric unsubtracted
Mandelstam integral $G_{2}(s,t)$ with positive double spectral function such
that%
\begin{equation}
A_{1}(s,t)=\lambda _{1}\left( \hat{R}(s,u)+\hat{R}(s,t)\right) +G_{1}(s,t)
\label{r21}
\end{equation}%
has partial wave amplitudes which satisfy the unitarity inequalities (\ref%
{int9}) for $s\geq 4$.

The crossed Regge term $\hat{R}(t,u)$ satisfies the assumptions of
Proposition \ref{drei}, see the bound (\ref{r20}) for $\hat{R}(s,u)$. Hence
we can find a constant $\lambda _{2}>0$ and a crossing symmetric amplitude $%
G_{2}(s,t)$ as indicated in Proposition \ref{drei} such that

\begin{equation}
A_{2}(s,t)=\lambda _{2}\hat{R}(t,u)+G_{2}(s,t)  \label{r22}
\end{equation}%
has partial wave amplitudes, which satisfy the unitarity inequalities (\ref%
{int9}) for $s\geq 4$.

To derive unitarity for the sum (\ref{r3}) the statement of Remark \ref%
{convex} about the inelastic inequalities is essential. Starting from (\ref%
{r21}) and (\ref{r22}) with constants $\alpha _{1,2}\geq 0$ such that $%
\alpha _{1}+\alpha _{2}\leq 1$ and $\alpha _{1}\lambda _{1}=\alpha
_{2}\lambda _{2}=c>0$ we obtain%
\begin{equation}
A(s,t)=\alpha _{1}A_{1}(s,t)+\alpha _{2}A_{2}(s,t)=c\left( \hat{R}(s,u)+\hat{%
R}(s,t)+\hat{R}(t,u)\right) +G(s,t),  \label{r23}
\end{equation}%
is an amplitude which satisfies the properties 1.-3. and 5. of section \ref%
{not}. Thereby $G(s,t)=\alpha _{1}G_{1}(s,t)+\alpha _{2}G_{2}(s,t)$ is given
by a Mandelstam representation without (if $\alpha (\infty )<0$) or with one
subtraction (if $\alpha (\infty )>0$ as assumed in Appendix \ref{unitRegge}%
). The constant $c$ can be absorbed into the residue function $\beta (t)$
and (\ref{r23}) yields the representation (\ref{r3}).

\section{Khuri poles\label{Khuri}}

\subsection{Mellin transformation\label{Mellin}}

Let $f(t)$ be a complex function with support in $\mathbb{R}_{+}$ such that $%
\int_{0}^{\infty }\left\vert f(t)t^{-\gamma }\right\vert ^{2}t^{-1}dt<\infty
$ exits for some $\gamma \in \mathbb{R}$. Then the Mellin transformation
\cite{Titchmarsh:1948}
\begin{equation}
a(\nu )=\frac{1}{\pi }\int_{0}^{\infty }f(t)t^{-\nu -1}dt=\mathcal{M}\left[
f(t)\right] (\nu )  \label{m1}
\end{equation}%
is defined at least for $\mathrm{Re}\nu =\gamma $ with $\int_{-\infty
}^{\infty }\left\vert a(\gamma +ix)\right\vert ^{2}dx=\frac{2}{\pi }%
\int_{0}^{\infty }\left\vert f(t)t^{-\gamma }\right\vert ^{2}t^{-1}dt$. The
inverse Mellin transformation is given by \cite{Titchmarsh:1948}
\begin{equation}
f(t)=\mathcal{M}^{-1}\left[ a(\nu )\right] (t):=\frac{1}{2i}\int_{\gamma
}a(\nu )t^{\nu }d\nu .  \label{m2}
\end{equation}%
The symbol $\int_{\gamma }d\nu $ means integration along the line $\nu
=\gamma +ix,\,-\infty <x<\infty $.

Let $\mathcal{L}_{\gamma },\,\gamma \in \mathbb{R}$, be the Hilbert space of
all functions $f(t)$ with a finite norm (\ref{el-7}), then the Mellin
transform maps this space isometrically onto the Sobolev space $\mathcal{S}%
(\gamma )$ of functions $a(\nu )$ with norm
\begin{equation}
\mid a(\nu )\mid _{\gamma }=\left[ \frac{\pi }{2}\int_{\mathbb{R}}dx\left(
\left\vert a(\gamma +ix)\right\vert ^{2}+\left\vert \frac{d}{dx}a(\gamma
+ix)\right\vert ^{2}\right) \right] ^{\frac{1}{2}}.  \label{m5}
\end{equation}%
If the support of $f(t)\in \mathcal{L}_{\gamma }$ lies inside $\left[
t_{0},\infty \right) $ with $t_{0}>0$, then the integral (\ref{m1}) exists
also for $\mathrm{Re}\,\nu >\gamma $ and defines a holomorphic function in
that region. Moreover we have $a(\nu )\in \mathcal{S}(\gamma ^{\prime })$
for all $\gamma ^{\prime }\geq \gamma $. If $f(t)\in \mathcal{L}_{\gamma }$
with $-1<\gamma <0$, then the dispersion integral $F(t)=\frac{1}{\pi }%
\int_{0}^{\infty }f(t^{\prime })(t^{\prime }-t)^{-1}dt^{\prime }$ exists,
and we can calculate the Mellin transformation of $F(-t)$
\begin{equation}
\varphi (\nu )=\mathcal{M}\left[ F(-t)\right] =\frac{1}{\pi }%
\int_{0}^{\infty }F(-t)t^{-\nu -1}dt  \label{m6}
\end{equation}%
with the result $\varphi (\nu )=-a(\nu )\left( \sin \pi \nu \right) ^{-1}$.
The inverse Mellin transform then yields the Khuri representation (\ref{rel4}%
) of the function $F(t)$%
\begin{equation}
F(t)=-\frac{1}{2i}\int_{\gamma }\frac{a(\nu )}{\sin \pi \nu }(-t)^{\nu }dv.
\label{m7}
\end{equation}%
If $a(\nu )\in \mathcal{S}(\gamma )$, where $\gamma $ is not an integer,
then the functions $\left( \sin \pi \nu \right) ^{-1}a(\nu )$ and \newline
$\left( \sin \pi \nu \right) ^{-1}\exp (\pm i\pi \nu )a(\nu )$ are also
elements of $\mathcal{S}(\gamma )$, and the mappings $a(\nu )\rightarrow
\left( \sin \pi \nu \right) ^{-1}a(\nu )$ and $a(\nu )\rightarrow \left(
\sin \pi \nu \right) ^{-1}\exp (\pm i\pi \nu )a(\nu )$ are continuous.

The following example is used in Appendix \ref{pole}. Let $t_{1}>0$ be a
positive number and $\alpha \in \mathbb{C}$ with $\mathrm{Re}\,\alpha >-1$.
Then the function $\mathbb{R}\ni t\rightarrow f(t)=\left( t^{\alpha
}-t_{1}^{\alpha +1}t^{-1}\right) \Theta (t-t_{1})$ is H\"{o}lder continuous,
and it is an element of $\mathcal{L}_{\gamma }$ for all $\gamma >\mathrm{Re}%
\alpha $. The Mellin transform $a(\nu )=\mathcal{M}\left[ f(t)\right] (\nu )$
is calculated%
\begin{equation}
a(\nu )=\frac{1}{\pi }\int_{t_{1}}^{\infty }\left( t^{\alpha }-t_{1}^{\alpha
+1}t^{-1}\right) t^{-\nu -1}dt=\frac{1}{\pi }t_{1}^{\alpha -\nu }\left(
\frac{1}{\nu -\alpha }-\frac{1}{\nu +1}\right) =\frac{1}{\pi }t_{1}^{\alpha
-\nu }\frac{\alpha +1}{(\nu -\alpha )(\nu +1)},  \label{m8}
\end{equation}%
if $\mathrm{Re}\,\nu >\mathrm{Re}\,\alpha $.

\subsection{The pole ansatz\label{pole}}

A Regge pole at $\nu =\alpha (s)$ with residue $\beta (s)$ leads to a series
of Khuri poles at positions $\nu =\alpha (s)-n,\,n=0,1,2,,...$ with residues
\begin{equation}
r_{n}(s)=\frac{(4-s)^{n}}{n!}\frac{(-\alpha )_{n}(-\alpha )_{n}}{(-2\alpha
)_{n}}\beta (s).  \label{k0}
\end{equation}%
A suitable Regge ansatz, which includes $N+1$ Khuri poles, is, see (\ref{m8}%
),
\begin{eqnarray}
a_{R}(s,\nu ) &=&a_{R}\left[ \alpha ,\beta \right] (s,\nu ):=\beta
\,t_{1}^{\alpha (s)-\nu }\left( \frac{1}{\nu -\alpha (s)}-\frac{1}{\nu +1}%
\right)  \nonumber \\
&&+\beta \,t_{1}^{\alpha (s)-\nu }\sum_{n=1}^{N}\frac{1}{n!}\left( \frac{4-s%
}{t_{1}}\right) ^{n}\frac{(-\alpha )_{n}(1+\nu )_{n}}{(1+n+2\nu )_{n}}\left(
\frac{1}{\nu -\alpha +n}-\frac{1}{\nu +1}\right) .  \label{k1}
\end{eqnarray}%
Thereby $t_{1}>16$ is a parameter. The function $\nu \rightarrow a_{R}(s,\nu
)$ is an element of the spaces $\mathcal{S}(\gamma )$ for all $\gamma >-1$
with $\gamma \notin \left\{ \mathrm{Re}\alpha (s)-n\mid n=0,1,...,N\right\} $%
. Moreover, its inverse Mellin transform is H\"{o}lder continuous in the
variable $t$.

Under the assumptions a) - c) of section \ref{R-el} about the trajectory $%
\alpha (s)$ it is sufficient to choose $N=2$ in (\ref{k1}). For large
energies, say $\left\vert s\right\vert >s_{2}$, we have $\mathrm{Re}\,\alpha
(s)<\gamma _{0},\,-1<\gamma _{0}<0$. Then the Regge amplitude is defined for
these energies by the integral (\ref{m7})
\begin{equation}
R(s,t)=-\frac{1}{2i}\int_{\gamma _{0}}\frac{a_{R}(s,\nu )}{\sin \pi \nu }%
(-t)^{\nu }dv  \label{k2}
\end{equation}%
with the $t$-channel absorptive part $R_{t}(s,t)=\mathcal{M}_{\gamma }^{-1}%
\left[ a_{R}(s,\nu )\right] $ if $t\geq 4$. We can shift the contour of
integration to $\mathrm{Re}\,\nu =\gamma _{1}$ and have to collect the
residues at $\nu =0,1,2$,%
\begin{equation}
R(s,t)=-\frac{1}{2i}\int_{\gamma _{1}}\frac{a_{R}(s,\nu )}{\sin \pi \nu }%
(-t)^{\nu }dv+a_{R}(s,0)+a_{R}(s,1)t+a_{R}(s,2)t^{2}.  \label{k3}
\end{equation}%
This formula is still correct when the Regge trajectory enters the region $%
\gamma <\mathrm{Re}\alpha (s+i0)<\frac{3}{2}$. Adding the crossed term $%
R(s,u)$ we obtain%
\begin{eqnarray}
R(s,t)+R(s,u) &=&-\frac{1}{2i}\int_{\gamma _{1}}\frac{a_{R}(s,\nu )}{\sin
\pi \nu }\left[ (-t)^{\nu }+(-u)^{\nu }\right] dv  \nonumber \\
&&+2a_{R}(s,0)+(4-s)a_{R}(s,1)+\left( t^{2}+u^{2}\right) a_{R}(s,2).
\label{k4}
\end{eqnarray}%
The sum $2a_{R}(s,0)+(4-s)a_{R}(s,1)+\left( t^{2}+u^{2}\right) a_{R}(s,2)$
may lead to singularities if a Regge pole crosses the integer values $\alpha
=0,1,2$. If $\alpha =0$ such a singularity appears in $a_{R}(s,0)$ unless
the residue vanishes. To compensate the pole at $\alpha =0$ the residue
function $\beta (s)$ has to include a ghost killing factor $\alpha (s)$ as
done in \cite{Kupsch:1977}. If $\alpha =1$ the Khuri pole at $\nu =1$ and
its daughter pole at $\nu =0$ compensate. This a kinematic pole killing due
to the projection onto even partial waves in (\ref{k4}). The case $\alpha =2$
is excluded by the assumptions on the trajectory function.

\begin{remark}
If a Regge pole enters the strip $\gamma _{0}+1<\mathrm{Re}\alpha
(s+i0)<\gamma _{0}+2$ an alternative formula for the Regge amplitude (\ref%
{k3}) is%
\begin{equation}
R(s,t)=-\frac{1}{2i}\int_{\gamma _{0}}\frac{a_{R}(s,\nu )}{\sin \pi \nu }%
(-t)^{\nu }dv+\frac{\pi r_{0}(s)}{\sin \pi \alpha (s)}(-t)^{\alpha (s)}-%
\frac{\pi r_{1}(s)}{\sin \pi \alpha (s)}(-t)^{\alpha (s)-1}.  \label{k5}
\end{equation}
\end{remark}

\noindent The Mellin transform of $\mathrm{Abs}_{t}A(s,t)$ is then given by%
\begin{equation}
a(s,\nu )=a_{R}\left[ \alpha ,\beta \right] (s,\nu )+b(s,\nu )  \label{k6}
\end{equation}%
where $b(s,\nu )$ is a holomorphic background. The crossed terms $%
R(t,s)+R(t,u)$ and $R(u,s)+R(u,t)$ do only contribute to the background $%
b(s,\nu )$ if the residue function $\beta (s)$ decreases at least like $%
\left\vert s\right\vert ^{-N-1}$.

\subsection{The unitarity integral\label{unitInt}}

If $A_{t}(s,.)\in \mathcal{L}_{\gamma _{0}}$ with $-\frac{1}{2}+\mu <\gamma
_{0}<0$ the Mellin transform $a(s,\nu )=\mathcal{M}\left[ A_{t}(s,t)\right] $
is holomorphic in $\nu $ for $\mathrm{Re}\,\nu >\gamma _{0}$, and it is an
element of $\cap _{\gamma \geq \gamma _{0}}\mathcal{S}(\gamma )$. For such
amplitudes the unitarity integral (\ref{el-11}) is transformed into \cite%
{Kupsch:1977}
\begin{equation}
w(s,\nu )=\sqrt{\frac{s-4}{s}}(s-4)^{\nu }I\left( \gamma _{0},\gamma
_{0};s,\nu \right) ,\,4\leq s\leq 16.  \label{ru1}
\end{equation}%
Thereby $w(s,\nu )=\mathcal{M}\left[ \psi (s,t)\right] $ is the Mellin
transform of the double spectral function, and $I\left( \gamma ,\gamma
^{\prime };s,\nu \right) $ is the Mellin-Barnes type integral%
\begin{equation}
I\left( \gamma ,\gamma ^{\prime };s,\nu \right) :=-\frac{1}{4\pi ^{2}}%
\int_{\gamma }d\xi \int_{\gamma ^{\prime }}d\eta ~(s-4)^{\xi +\eta -2\nu
}M(s,\nu ,\xi ,\eta )a(s+i0,\xi )a(s-i0,\eta ).  \label{ru2}
\end{equation}%
with the kernel
\begin{equation}
M(s,\nu ,\xi ,\eta )=B(1+\xi ,\nu -\xi )B(1+\eta ,\nu -\eta )B(1+\nu ,1-\nu
+\xi +\eta )  \label{ru3}
\end{equation}%
The function $B(x,y)=\Gamma (x)\Gamma (y)\left( \Gamma (x+y)\right) ^{-1}$
is the Euler beta function. The integral (\ref{ru4}) is defined with
integration along the lines $\mathrm{Re}\,\xi =\gamma $ and \textrm{Re}%
\thinspace $\eta =\gamma ^{\prime }$ such that \newline
$\gamma ,\gamma ^{\prime }<\mathrm{Re}\,\nu <1+\gamma +\gamma ^{\prime }$.
If $a(s,\nu )$ is holomorphic for $\mathrm{Re}\,\nu >\gamma \geq \gamma _{0}$%
, the unitarity identity (\ref{ru1}) has an analytic continuation to
\begin{equation}
w(s,\nu )=\sqrt{\frac{s-4}{s}}(s-4)^{\nu }I\left( \gamma ,\gamma ;s,\nu
\right) ,\,4\leq s\leq 16,  \label{ru4}
\end{equation}%
which is valid within the strip $\gamma <\mathrm{Re}\,\nu <1+2\gamma $. Take
$\nu $ in the strip \newline
$\gamma <\mathrm{Re}\,\nu <\min \left\{ \gamma +1,1+2\gamma \right\} $. By a
shift of the contour of integration to the right to $\gamma +1$ (such that $%
\gamma <\mathrm{Re}\,\nu <\gamma +1$) we obtain%
\begin{eqnarray}
&&\sqrt{\frac{s}{s-4}}(s-4)^{-\nu }w(s,\nu )=B(1+\nu ,1+\nu )a(s+i0,\nu
)a(s-i0,\nu )  \nonumber \\
&&\hspace{1.0in}+a(s+i0,\nu )\varphi _{1}(s-i0,\nu )+\varphi _{1}(s+i0,\nu
)a(s-i0,\nu )  \nonumber \\
&&\hspace{1.0in}+I\left( \gamma +1,\gamma +1;s,\nu \right)  \label{ru5}
\end{eqnarray}%
with%
\[
\varphi _{1}(s\pm i0,\nu )=\frac{1}{2\pi i}\int_{\gamma +1}d\xi \,(s-4)^{\xi
-\nu }B(1+\xi ,\nu -\xi )\,B(1+\nu ,1+\xi )\,a(s\pm i0,\xi ).
\]%
For $4\leq s\leq 16$ we have $w(s,\nu )=(2i)^{-1}\left( a(s+i0,\nu
)-a(s-i0,\nu )\right) $. Assume that a pole $\beta (s+i0)\left( v-\alpha
(s+i0)\right) ^{-1}$ of $a(s+i0,\nu )$ enters the strip $\gamma <\mathrm{Re}%
\,\nu <\gamma +1$, then at the residue the following identity follows from (%
\ref{ru5})
\begin{equation}
2i\sqrt{\frac{s-4}{s}}(s-4)^{\alpha (s+i0)}\varphi (s-i0,\alpha (s+i0))=1,
\label{ru6}
\end{equation}%
where%
\begin{equation}
\varphi (s\pm i0,\nu )=B(1+\nu ,1+\nu )a(s\pm i0,\nu )+\varphi _{1}(s\pm
i0,\nu )  \label{ru7}
\end{equation}%
is (up to a factor $2$) the reduced partial wave amplitude of $A(s\pm i0,t)$%
. If $\left\vert v-\alpha (s-i0)\right\vert <1$ we have%
\begin{equation}
\varphi (s-i0,\nu )=\beta (s-i0)B(1+\nu ,1+\nu )\left( v-\alpha
(s-i0)\right) ^{-1}+\phi (s-i0,\nu ),  \label{ru8}
\end{equation}%
with a holomorphic \textquotedblleft background\textquotedblright\ $\phi
(s-i0,\nu )$, which originates from the daughter pole contributions, the
background of $a(s-i0,\nu )$ and from $\varphi _{1}(s-i0,\nu )$. Since $%
\mathrm{Im}\,\alpha $ is small, the number $\alpha (s+i0)$ lies in the
neighbourhood of $\alpha (s-i0)$, and we can insert (\ref{ru8}) into (\ref%
{ru6}). The resulting identity
\begin{equation}
\beta (s-i0)=(s-4)^{\sigma -\alpha (s+i0)}B^{-1}\chi (s)-2i\sqrt{\frac{s-4}{s%
}}(s-4)^{\sigma }\chi (s)B^{-1}\phi (s-i0,\alpha )  \label{ru9}
\end{equation}%
with $B^{-1}=B^{-1}(1+\alpha ,1+\alpha )$ has exactly the form (\ref{rel11}%
), only the interpretation of $b(s,\nu )$ has changed. The function $\mathbb{%
R}\ni s\rightarrow B^{-1}(1+\alpha (s+i0),1+\alpha (s+i0))\phi (s-i0,\alpha
(s+i0))$ is H\"{o}lder continuous.

\end{document}